\documentclass[12 pt]{article}
\pdfoutput=1
\usepackage[utf8]{inputenc}
\usepackage[breaklinks]{hyperref}%for \url{..}
\usepackage{amssymb}
\usepackage{amsmath}
\usepackage{amsthm}
\usepackage{stmaryrd}
\usepackage{bm}%boldmath
\usepackage{complexity}
\usepackage[a4paper, total={5in, 8in}]{geometry}
\usepackage[textsize=scriptsize]{todonotes}
\usepackage{booktabs}%better tables
\usepackage{tikz}
\usetikzlibrary{positioning, arrows}
\usetikzlibrary{patterns}
\usetikzlibrary{shapes}
\usetikzlibrary{decorations.pathreplacing}
\usetikzlibrary{fadings}
\usetikzlibrary{intersections}
\usepackage{xifthen}
\usepackage{pgfmath}
\usepackage{algorithm,algorithmic}
\usepackage{mathtools}

\newcommand{\h}{\bm{h}}
\newcommand{\w}{\bm{w}}
\renewcommand{\o}{\bm{o}}
\newcommand{\Nat}{\mathbb{N}}

\newcommand{\eqdef}{\stackrel{\text{def}}{=}}

\newcommand{\leqaug}{\leq_{\text{aug}}}
\newcommand{\leqstruct}{\leq_{\text{st}}}
\newcommand{\geqaug}{\geq_{\text{aug}}}
\newcommand{\geqstruct}{\geq_{\text{st}}}
\newcommand{\Bad}{{\textit{Bad}}}
\newcommand{\Dec}{{\textit{Dec}}}%decreasing
\newcommand{\Inc}{{\textit{Inco}}}%incomparable
\newcommand{\shift}{{\textit{shift}}}
\newcommand{\om}{{\omega}}
\newcommand{\omom}{{\om^{\om}}}
\newcommand{\nfp}{ \mathbin{\overset{.}{+}}}

\title{On the cartesian product of well-orderings}
\author{Isa Vialard\footnote{Université Paris-Saclay, ENS Paris-Saclay, CNRS, LMF, Gif-sur-Yvette, France}}

\theoremstyle{plain}%for theorems and the like
\newtheorem{theorem}{Theorem}[section]
\newtheorem{lemma}[theorem]{Lemma}
\newtheorem{proposition}[theorem]{Proposition}
\newtheorem{corollary}[theorem]{Corollary}

\newtheorem{example}[theorem]{Example}
\theoremstyle{definition}%for definitions and the like
\newtheorem{definition}[theorem]{Definition}
\newtheorem{remark}[theorem]{Remark}

\begin{document}

\maketitle
%\begin{center}
%\fbox{\textbf{ Draft version of \today. Do not circulate}} 
%\end{center}

\begin{abstract}
The width of a well partial ordering (wpo) is the ordinal rank of the set of its antichains ordered by inclusion. We compute the width of wpos obtained as cartesian products of finitely many well-orderings.
\end{abstract}

\section{Introduction}

For a finite poset, and more generally a finite quasi-order (qo), there are intuitive notions of dimension that play a paramount role in combinatorics and algorithmics: its cardinal, but also its height (the cardinal of its longest chain) and its width (the cardinal of its longest antichain, i.e., sequence of incomparable elements).

 With some provisions, these dimensions can be extended to infinite posets:
 If a qo is \emph{well-founded} (or WF), we can define its \emph{height} as the rank of the tree of strictly decreasing sequences.
If a qo is FAC (it only has finite antichain), we can define its \emph{width} as the rank of the trees of antichains.
If a qo is both WF and FAC, then it is called a \emph{well quasi-order}(wqo)(\cite{higman52}): it has a height, a width, and also a maximal order type (\cite{dejongh77}), defined as the rank of the trees of bad sequences (a sequence
$x_0 , x_1 , x_2 , \dots$ is \emph{good} if there are some positions $i < j$ such that $x_i \leq x_j$, and \emph{bad} otherwise).
Wqos can alternatively be defined as qos that do not have infinite bad sequences.

There is a rich theory of wqos (\cite{milner85},\cite{schuster2020}), where these dimensions, which we call \emph{ordinal invariants}, are used to measure complexity. De Jongh and Parikh (\cite{dejongh77}) and Schmidt (\cite{schmidt79}) initiated the study of the maximal order type, for use in proof theory. K\v{r}\'i\v{z} and Thomas (\cite{kriz90b}) later introduced the width for infinitary combinatorics. Blass and Gurevich (\cite{blass2008}) then contributed to the study of invariants for program verification. The maximal order type was also applied in \cite{bonnet2013} for expressiveness results.

In the study of \emph{well-structured transition systems} (WSTS), i.e., systems whose set of configurations is a wqo, some upper bound results on complexity rely on the length of controlled bad sequences of configurations (\cite{HSS-lmcs},\cite{HSS-lics2012}), i.e., on the maximal order type of the underlying wqo. In \cite{schmitz2019b}, Schmitz refined this technique with controlled antichains whose length depends on the width instead. 

A recent article by Dzamonja et al. (\cite{dzamonja2020}) shows that we do not always know how to compute the width of wqos, even in the apparently simple case of a cartesian product. This problem is unfortunate since the cartesian product is the most common and basic data structure in mathematics and computer science. However, one subproblem, the special case of the width of the cartesian product of two linear well-founded orders, i.e., two ordinals, was solved by  Abraham (\cite{abraham87}) 35 years ago.

This article develops a method to compute the width of the cartesian product of $n$ ordinals, for any $n\in\mathbb{N}$. As explained in Section \ref{sec-abraham}, the method of residuals used in \cite{abraham87} relies on specificities of the case $n=2$, which are lost in case $n=3$ and beyond. Our method consequently develops a more general, game-theoretical approach.

\subsection{Outline of the article}

Section \ref{section-preliminaries} introduces definitions, notations and recalls known results, mostly following \cite{dzamonja2020}. Section \ref{section-lowerbound} proves intermediary results on lower bounds on the width that lay the ground for future proofs.

\medskip

The following three sections gradually progress toward our main result:
In Section \ref{section-indecomposable}, we compute the width of the product of indecomposable ordinals.
Section \ref{section-infinite} uses the previous result to compute the width of the product of infinite ordinals.
Section \ref{section-finite-infinite} eventually extends this result by adding finite ordinals to the product.

\medskip
In Section \ref{section-w-equals-o} we leverage our main result to compute the width of the cartesian product of more general wqos.

For completeness, Section \ref{section-finite} recalls classic results for the cartesian product of finite ordinals.

\section{Preliminaries}
\label{section-preliminaries}

We assume familiarity with basics of order theory. See e.g.\ \cite{fraisse86}.

\subsection{The measure of wqos}

\medskip
For any wqo $(A,\leq_A)$ (we often write just $A$ when the order is understood), we write $\Inc(A)$ (resp. $\Dec(A)$ and $\Bad(A)$)  for the tree of non-empty antichains (resp. strictly decreasing sequences, bad sequences) in $A$ ordered by initial segment: if $s = (x_1,\dots,x_n)$ and $t = (x_1,\dots,x_n,y)= s\frown y$, then $t$ is a child of $s$.

Observe that, since $A$ is a wqo (so FAC and WF), the trees $\Inc(A)$, $\Dec(A)$ and $\Bad(A)$ do not have infinite branches: they are well-founded. However, they can be infinitely branching.

\medskip
One can ascribe a \emph{rank} to any node of a well-founded tree $T$ from bottom to top. Let $x \in T$  a node: if $s$ is a leaf, then $r(s) =0$. Else $r(s) = \sup \{r(t) +1\; |\; t \text{ is a son of } s\}$.
The rank of $T$ is defined as the image of $T$ through $r$, i.e., $r(T) = sup_{s\in T} \;r(s) +1$. Since $T$ can be infinitely branching, its rank is a possibly infinite ordinal.
For a more general definition of the rank of any well-founded partial-order set, see Section 2.2 of \cite{dzamonja2020}.

\begin{definition}
\label{def-who}
The \emph{width} $\w(A)$, the \emph{height} $\h(A)$, and the \emph{maximal order type} $\o(A)$ are respectively the rank of $\Inc(A)$, $\Dec(A)$, and $\Bad(A)$. Together, they are called the \emph{ordinal invariants} of $A$.
\end{definition}

\begin{example}

 For any ordinal $\alpha>0$ with order $\subset$, $\o(\alpha)=\h(\alpha)=\alpha$, and $\w(\alpha)=1$.

\end{example}

\begin{remark} The trees $\Inc(A)$, $\Dec(A)$, and $\Bad(A)$ are, more precisely, forests, since they have multiple roots: all one-element sequences.
\end{remark} 
%If that is confusing to you, you can see the empty sequence as only root, and take $r(\emptyset)$ as the rank of the tree. This definition is equivalent.

\medskip
Since an antichain is a bad sequence, $\Inc(A)$ is a subtree of $\Bad(A)$. Hence:
\begin{lemma} 
\label{lem-w-leq-o}
For all wqo $A$, $\w(A) \leq \o(A)$.
\end{lemma}

\begin{remark}
The maximal order type 
$\o(A)$ was historically defined as the maximal linearisation of $A$ (\cite{dejongh77}), i.e., a linear order (an ordinal) $(A,\leq)$ such that $\leq_A\subseteq \leq$. The height $\h(A)$ can similarly be defined as the longest chain of $A$. Nonetheless, Definition \ref{def-who} allow us to include the width as a third natural invariant.
\end{remark}

\medskip 

When computing ordinal invariants, we frequently want to compare a wqo to one of its substructures or augmentations:

Let $(A,\leq_A),(B,\leq_B)$ be two wqos.
$A$ is an \emph{augmentation} of $B$ if the sets $A,B$ are equal and $\leq_B \subseteq \leq_A$. We denote it by $A \geqaug B$.
$A$ is a \emph{substructure} of $B$ if the set $A$ is a subset of $B$ and 
$\leq_A = (\leq_B)_{|A}$. We denote it by $A \leqstruct B$. 
We use the notation $A\equiv B$ when $A$ is isomorphic to $B$, i.e., when there is a bijection between $A$ and $B$ that preserves the order. We often abuse terminology and say that $A$ is a substructure (resp. an augmentation) when $A$ is isomorphic to a substructure (resp. an augmentation).

\begin{example}
For any ordinals $\alpha < \beta$, $\alpha \leqstruct \beta$.
\end{example}

\begin{example}
For any $n \in\mathbb{N}$, $\Gamma_n \leqaug n$, where $\Gamma_n$ is the antichain with $n$ elements.
\end{example}

We define the disjoint sum $\sqcup$, lexicographic sum $+$, cartesian product $\times$ and direct product $\cdot$ as in \cite{dzamonja2020}.

\begin{example}
For any wqos $A,B$, $A \sqcup B \leqaug A + B$, and $A \times B \leqaug A \cdot B$ 
\end{example}

Observe that, for two wqos $A$ and $B$, if $A \leqstruct B$, then $\Inc(A)$,$\Dec(A)$, and $\Bad(A)$ are respectively subtrees of $\Inc(B)$,$\Dec(B)$, and $\Bad(B)$. Similarly, if $A \geqaug B$, then all antichains or bad sequences of $A$ are antichains or bad sequences of $B$. However, strictly decreasing sequences of $B$ are strictly decreasing in $A$ whenever $A$ and $B$ are partial orders. Therefore:
\begin{lemma}
\label{lem-leq-aug}
For all wqos $A,B$,
\begin{itemize}
\item If $A \leqstruct B$, then $*(A) \leq *(B)$ for $* = \w,\o,\h$.
\item If $A \geqaug B$, then $\*(A) \leq \*(B)$ for $* = \w,\o$. If furthermore $A, B$ are wpos, $\h(A) \geq \h(B)$.
\end{itemize}
\end{lemma}

\subsection{Residual Characterization}

For a quasi-order $A$, $x \in A$, and a relation symbol $* \in \{\perp, <, >, \not\leq, \not\geq \}$, we define the \emph{$*$-residual of $A$ at $x$} as
$$A_{* x} = \{y \in A : y * x\} \;.$$
We can generalize this notion to subsets $Y\subseteq A$:
$$A_{* Y} = \bigcap_{x\in Y} A_{*x} \;.$$

If $Y=\emptyset$, $A_{*Y}=A$.
A residual can be seen as a substructure of $A$. Thus it is a wqo, with its own ordinal invariants, smaller than or equal to the ordinal invariants of $A$ (see Lemma \ref{lem-leq-aug}).

For instance, $\mathbb{N}_{<2} = \{0,1\}$ and $\mathbb{N}_{\perp 2} = \emptyset$.
In Figure \ref{fig-residu-Nat2}, you can see the residuals of $\mathbb{N}^2$ at $x=(4,6)$ in colors:
$\mathbb{N}^2_{<x}$ in blue, $\mathbb{N}^2_{>x}$ in green, and $\mathbb{N}^2_{\perp x}$ in red. The union of the red and blue parts is $\mathbb{N}^2_{\not\geq x}$, the union of the red and green parts is $\mathbb{N}^2_{\not\leq x}$.
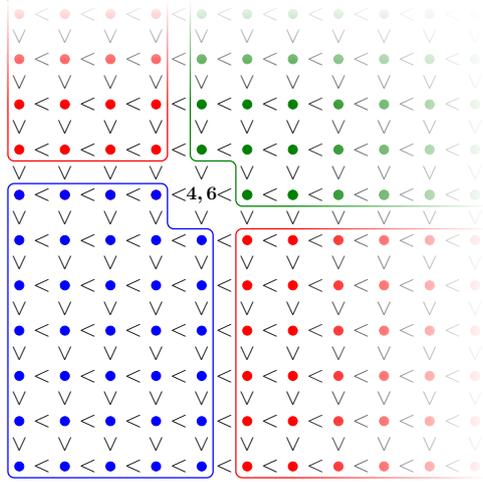
\begin{figure}[ht]
\centering
\scalebox{0.6}{
\tikzfading[name=fade right, left color=transparent!100, right color=transparent!0]
\tikzfading[name=fade up, bottom color=transparent!100, top color=transparent!0]
\begin{tikzpicture}[rounded corners]
\draw[use as bounding box,white] (-0.32,-0.32) rectangle (10.32,10.32);

%  was helpful \draw[help lines] (-0.5,-0.5) grid (11,11);

\clip (-0.32,-0.32) -- (-0.32,10.32) -- (10.32,10.32) -- (10.32,-0.32) -- cycle; % Clips the picture...

\foreach \x [evaluate=\x as \xval] in {0,...,10}{
  \foreach \y in {0,...,10}{
    % CAS1 : x<4
    \ifnum \x<4
    {\ifnum \y<6
       \node[draw,circle,inner sep=2pt,fill,blue] at (\x,\y) {};
     \else
       \node[draw,circle,inner sep=2pt,fill,red] at (\x,\y) {};
    \fi}
    \fi
    % CAS2 : x>4
    \ifnum \x>4
    {\ifnum \y>5
      \node[draw,circle,inner sep=2pt,fill,green!50!black] at (\x,\y) {};
    \else
      \node[draw,circle,inner sep=2pt,fill,red] at (\x,\y) {};
    \fi}
    \fi
    %  CAS3: la colonne x=4
    \ifnum \x=4
      \ifnum \y<6
	\node[draw,circle,inner sep=2pt,fill,blue] at (\x,\y) {};
      \fi
      \ifnum \y>6
	\node[draw,circle,inner sep=2pt,fill,green!50!black] at (\x,\y) {};
      \fi
    \fi
    % CAS4: la ligne y=6
    \ifnum \y=6
      \ifnum \x<4
	\node[draw,circle,inner sep=2pt,fill,blue] at (\x,\y) {};
      \fi
      \ifnum \x>4
	\node[draw,circle,inner sep=2pt,fill,green!50!black] at (\x,\y) {};
      \fi
    \fi
    %% Enfin La relation entre deux noeuds voisins immédiats:
    \node at (\x+0.5,\y) {\large $<$};
    \node[rotate=90] at (\x,\y+0.5) {\large $<$};
  }
}

\node[fill=white,inner sep=0pt] at (4,6) {$\bm{4,6}$};

% On entoure les résidus:
% d'abord le bloc NE des supérieurs, en vert foncé
\draw[thick,green!50!black] (3.75,11) -- (3.75,6.75) -- (4.75,6.75) -- (4.75,5.75) -- (11,5.75);
% puis le bloc SW des inférieurs, en bleu
\draw[thick,blue] (-0.25,-0.25) -- (-0.25,6.25) -- (3.25,6.25) -- (3.25,5.25) -- (4.25,5.25) -- (4.25,-0.25) -- cycle;
% puis les deux blocs disjoints des incomparables, en rouge
\draw[red,thick] (12,-0.25) -- (4.75,-0.25) -- (4.75,5.25) -- (12,5.25); % right part
\draw[red,thick] (-0.25,12.25) -- (-0.25,6.75) -- (3.25,6.75) -- (3.25,12.25); % left part

% FADING EFFECT to right and to top
\fill [path fading=fade right,white] (6.0,-0.32) rectangle (10.32,10.32);
\fill [path fading=fade up,white] (-0.32,8.0) rectangle (10.32,10.32);

\end{tikzpicture}
}%scalebox
\caption{$\Nat^2$: residuals by colors.}
\label{fig-residu-Nat2}
\end{figure}
%% Local Variables:
%% fill-column: 999
%% End:

\medskip
Residuals are essential in the computation of invariants, given:
\begin{align}
\w(A) &= \sup_{x\in A} \bigl(\w(A_{\perp x}) + 1\bigr) \\
\h(A) &= \sup_{x\in A} \bigl(\h(A_{< x}) + 1\bigr) \\
\o(A) &= \sup_{x\in A} \bigl(\o(A_{\not\geq x}) + 1\bigr) \label{eq-residual-o}
\end{align}

These formulas can be seen as a reformulation of tree rank computation (see Section 2.3. of \cite{dzamonja2020} and the references therein). 
We can use them to recursively compute the invariants of $A$: this is called the \emph{method of residuals}.

\begin{example}
For all $x\in \mathbb{N}^2$, $\mathbb{N}^2_{< x}$ is finite  and contains arbitrarily many elements (see Figure \ref{fig-residu-Nat2}). Therefore $\h(\mathbb{N}^2)= \om$.
\end{example}
\begin{remark}
In \cite{dejongh77}, where $\o(X)$ is defined as the maximal linearisation of $X$, Equation (\ref{eq-residual-o}) is an essential theorem since it allows to prove that this definition of the maximal order type is equivalent to its characterization as the rank of $\Bad(X)$.
\end{remark}

\subsection{State of the art}

The state of the art on ordinal invariants can be found in \cite{dzamonja2020}. Here are some useful results:

\begin{lemma}[Lemma 4.1 from \cite{dzamonja2020}]
\label{lem-DSS-sum}
Let $\{P_i : i < \alpha\}$ be an $\alpha$-indexed family of wqos.
Then $\sum_{i<\alpha} P_i$, the lexicographic sum along the ordinal $\alpha$ (see definition in \cite{dzamonja2020}), is a wqo, and $\w(\sum_{i<\alpha} P_i) = sup_{i<\alpha} \w(P_i)$.
\end{lemma}

\begin{lemma}[Lemma 4.2 from \cite{dzamonja2020}]
\label{lem-DSS-sqcup}
Let $A,B$ be wqos. Then $A\sqcup B$ is a wqo, and 
$\w(A \sqcup B ) = \w(A) \oplus \w(B)$, where $\oplus$ is the natural addition on ordinals.
\end{lemma}

From Lemma \ref{lem-DSS-sqcup} we can deduce the following result:

\begin{lemma}
\label{lem-times-Gamma}
Let $A$ be a wqo, and $n\in \mathbb{N}$. Then $\w(A\times \Gamma_n) = \w(A) \otimes n$, where $\otimes$ is the natural multiplication on ordinals.
\end{lemma}
\begin{proof}
Observe that $A \times \Gamma_n \equiv A \sqcup \dots \sqcup A$ the disjoint sum of $n$ copies of $A$. Therefore $\w(A\times \Gamma_n) = \w(A)\oplus\dots\oplus \w(A) = \w(A) \otimes n$.
\end{proof}

The maximal order type and the height of a cartesian product $A\times B$ are functional in the maximal order type of $A$ and $B$, and their height respectively:

\begin{lemma}[Theorem 3.5 of \cite{dejongh77}]
\label{lem-product-o}
Let $A,B$ be wqos. Then $A\times B$ is a wqo, and 
$\o(A \times B ) = \o(A) \otimes \o(B)$.
\end{lemma}

\begin{lemma}[Lemma 4.6 of \cite{dzamonja2020}, see \cite{abraham87} for the proof]
\label{lem-product-h}
Let $A,B$ be wqos. Then 
$\h(A \times B ) = \sup\{\alpha \oplus\beta + 1 \;|\; \alpha < \h(A),  \beta < \h(B)\}$.
\end{lemma}

However, the width of $A\times B$ is not functional in any of the invariants of $A$ and $B$:
\begin{example}
Let $H=\sum_{n<\om}\Gamma_n$. Observe that $\w(H)=\o(H)=\h(H)=\om$. Let $A_1 = H + H$ and $A_2= H + \om$ (see Figure \ref{fig-H}). With the method of residuals, one can see that $A_1$ and $A_2$ have the same invariants: $\w(A_1)=\w(A_2)=\om$, and $\o(A_1)=\o(A_2)=\h(A_1)=\h(A_2)=\om\cdot 2$. However $\w(A_1\times \om)= \om^2 \cdot 2 \neq \w(A_2 \times \om)= \om^2 + \om$.

\begin{figure}
\centering
    \scalebox{.8}{\begin{tikzpicture}[rounded corners]
\begin{centering}

\scalebox{1}{
%\draw[use as bounding box,white] (-1.4,-0.32) rectangle (14.1,10.85);

% was helpful \draw[help lines] (-0.5,-0.5) grid (11,11);

\foreach \y [evaluate=\y as \yval] in {0,...,8}{
  \foreach \x in {0,...,\yval}{

    \node[draw,circle,inner sep=1pt,fill,color=red] at (\x/2-\yval/4+2,\y/2) {};
    \node[draw,circle,inner sep=1pt,fill,color=blue] at (\x/2 -\yval/4 +8,\y/2) {};
    \node[draw,circle,inner sep=1pt,fill,color=red] at (\x/2 -\yval/4+2,\y/2+4.5) {};
    
  }
}

\foreach \y in {0,...,8}{

    \node[draw,circle,inner sep=1pt,fill,color=blue] at (8,4.5+\y/2) {};

}

\fill [path fading=fade up,white] (-1, .5) rectangle (11,4.4);
\fill [path fading=fade up,white] (-1, 5) rectangle (11,9);

\node [color=red] at (-1,4.5){$ A_1\;:$};
\node [color=blue] at (5,4.5){$A_2\;:$};
\node at (2,7){$H$};
\node at (2,2.5){$H$};
\node at (8,2.5){$H$};
\node at (7.7,7){$\om$};

}

\end{centering}
\end{tikzpicture}}

    \caption{Two wqos $A_1$ and $A_2$ with the same invariants such that $\w(A_1\times\om)\neq\w(A_2\times \om)$.}
\label{fig-H}
\end{figure}
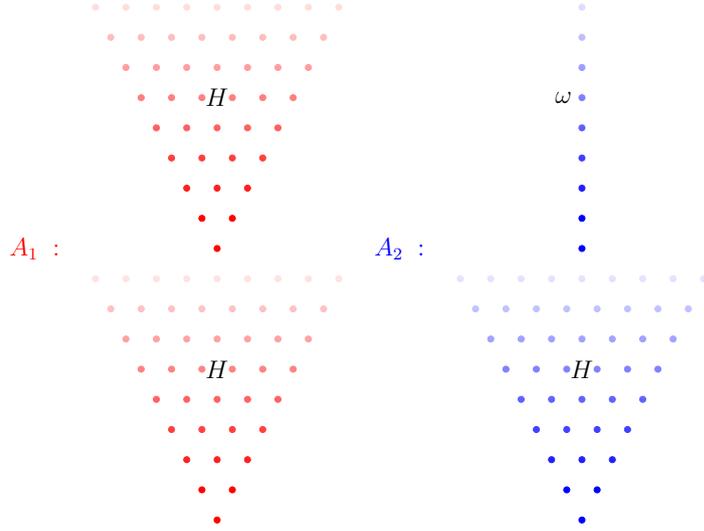

\end{example}

\subsection{Cartesian product of two ordinals}
\label{sec-abraham}

Abraham (\cite{abraham87}) used the method of residuals to compute the width of the cartesian product of two ordinals. Let us recall the main steps of his proof:

Let $\alpha,\beta$ be two ordinals. According to the method of residuals, 
\begin{align}
\label{residual-two-ordinals}
w(\alpha\times\beta) = \sup_{(x_1,x_2)\in \alpha\times\beta} \w((\alpha\times\beta)_{\perp (x_1,x_2)}) + 1\;.
\end{align}

Fix $(x_1,x_2) \in \alpha\times\beta$. Then for any $(y_1,y_2)\in\alpha\times\beta$, $(x_1,x_2)\perp(y_1,y_2)$ iff $x_1<y_1$ and $x_2>y_2$, or $x_1>y_1$ and $x_2<y_2$. Thus the residual $(\alpha\times\beta)_{\perp (x_1,x_2)}$ is a disjoint union (see Figure \ref{fig-res-2}): $$(\alpha\times\beta)_{\perp (x_1,x_2)} = \alpha_{<x_1}\times \beta_{>x_2} \;\sqcup\; \alpha_{>x_1}\times \beta_{<x_2}\;.$$

  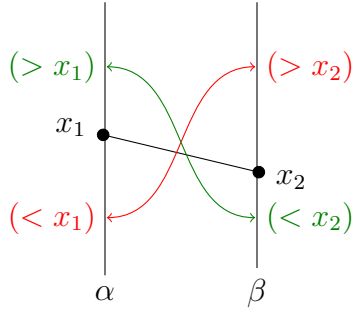
\begin{figure}[h]
            
               \begin{centering}
                \begin{tikzpicture}
                    \node (S) at (0,0){};
                    \node (S') at (2,0){};
                    \node (D) at (0,-4){$\alpha$};
                    \node (D') at (2,-4){$\beta$};
                    \node (X) at (-.45,-1.8){$x_1$};
                    \node (Y) at (2.45,-2.5){$x_2$};
                    \node (>X)[green!50!black] at (-.7,-1){$(> x_1)$};
                    \node (<X)[red] at (-.7,-3){$(< x_1)$};
                    \node (>Y)[red] at (2.7,-1){$(> x_2)$};
                    \node (<Y)[green!50!black] at (2.7,-3){$(< x_2)$};
                    \path[-] (S) edge node {} (D);
                    \path[-] (S') edge node {} (D');
                    \path[*-*] (X) edge node {} (Y);
                    \path[<->][red] (<X) edge [in = 180, out=0] node {} (>Y);
                    \path[<->][green!50!black] (>X) edge [in = 180, out=0] node {} (<Y);
               \end{tikzpicture}
               %remove the next empty line and centering doesnt work. latex sucks.

               \end{centering}
               \caption{Residual of $\alpha\times\beta$ at $(x_1,x_2)$ as a disjoint union.
            }
            \label{fig-res-2}
        \end{figure}

 Observe that $\alpha_{<x_1}$ is isomorphic to $x_1$, and $\alpha_{>x_1}$ to $\alpha - (x_1 +1)$. The same reasoning applies to $\beta_{<x_2}$ and $\beta_{>x_2}$. Using
 Lemma \ref{lem-DSS-sqcup} we can rewrite Equation (\ref{residual-two-ordinals}) as:  $$\w(\alpha\times\beta) = \sup_{x_1\in{\alpha}\atop x_2\in\beta} \bigl(\w(x_1\times(\beta -x_2)) \oplus \w((\alpha - x_1) \times x_2)\bigr) + 1\;.$$

This is how $\w(\alpha\times\beta)$ is computed in \cite{abraham87}. Here are the main results (the successor and limit cases are dealt with separately):

\begin{lemma}[Lemma 3.2 of \cite{abraham87}]
If $\alpha$ is infinite, then $\w(\alpha\times(\beta+ 1)) = \w(\alpha\times \beta) + 1$. (If $\alpha$ is finite then
the equality holds iff $\alpha \geq \beta + 1$).
\end{lemma}

Any ordinal can be written in Cantor normal form as $\alpha = \om^{\alpha'}\cdot a \nfp \rho$, where
the symbol $\nfp$ is used to point out a $+$ that could be replaced by a $\oplus$. Similarly $\alpha$ can be written in Cantor normal form without multiplicities $\alpha = \sum_{i \in [0,l]} \om^{\alpha_i}$, with $\alpha_0\geq\dots\geq\alpha_l$. We will use the latter in Section \ref{section-infinite}.

\begin{theorem}[Theorem 3.4 of \cite{abraham87}]
\label{thm-abraham-limit}
For any ordinals $\alpha = \om^{\alpha'}\cdot a \nfp \rho$ and $\beta= \om^{\beta'}\cdot b \nfp \sigma$ ,
$$\w(\om\alpha\times\om\beta)=\om\om^{\alpha'\oplus\beta'} \nfp [\w(\om\om^{\alpha'}\times\om\sigma) \oplus \w(\om\om^{\beta'}\times\om\rho)]\;.$$

\end{theorem}
With a simple change of variables, this becomes $$\w(\alpha\times \beta) =
\om^{\eta} \nfp [\w(\om^{\alpha'}\times\sigma) \oplus \w(\om^{\beta'}\times\rho)]$$ if $\alpha,\beta$ are limit ordinals, with $\eta = 1+(\alpha'-1)\oplus(\beta'-1)$.

\bigskip

Now let us try to use the method of residuals for the product of $n\geq 2$ ordinals $\alpha_1,\dots,\alpha_n$. Let $X = \alpha_1\times\dots\times\alpha_n$ and $x,y \in X$. Then $x=(x_1,\dots,x_n)\perp y=(y_1,\dots,y_n)$ iff there exist $i,j \in [1,n]$ such that $x_i < y_i$ and $x_j > y_j$. With some work, we can express the residual $X_{\perp x}$ as the union of subsets of the form $\underset{i\in I_1}{\bm{\times}} (<x_i) \times \underset{i\in I_2}{\bm{\times}} (>x_i) \underset{i\not\in I_1\cup I_2}{\bm{\times}} \{x_i\} $ with $I_1,I_2\subseteq [1,n]$ disjoint and non-empty.

 However, unlike the case $n=2$, this union of subsets is not a disjoint union of wqos. Take for instance $n=3$ and the subsets $(>x_1)\times (>x_2) \times (<x_3)$ and $(>x_1)\times (<x_2) \times (<x_3)$ (see Fig. \ref{fig-residual-fail-3}): they can have comparable elements. We could say that the residual is an augmentation of a disjoint union, but this method can only give us an upper bound on $\w(X)$.
 
 \begin{figure}[h]

               \begin{centering}
                \begin{tikzpicture}
                    \node (S) at (0,0){};
                    \node (S') at (2,0){};
                    \node (S'') at (4,0){};
                    \node (D) at (0,-4){$\alpha_1$};
                    \node (D') at (2,-4){$\alpha_2$};
                    \node (D'') at (4,-4){$\alpha_3$};
                    \node (X) at (-.45,-1.8){$x_1$};
                    \node (Y) at (2.25,-2.5){};
                    \node (Yl) at (2.4,-2.1){$x_2$};
                    \node (Y') at (1.8,-2.45){};
                    \node (Z) at (4.45,-2){$x_3$};
                    \node (>Y)[red] at (3,-0.8){$(>x_1)\times(> x_2)\times(<x_3)$};
                    \node (>Y)[green!50!black] at (1,-3.5){$(>x_1)\times(< x_2)\times(<x_3)$};
                    \path[-] (S) edge node {} (D);
                    \path[-] (S') edge node {} (D');
                    \path[-] (S'') edge node {} (D'');
                    \path[*-*] (X) edge node {} (Y);
                    \path[-*] (Y') edge node {} (Z);
                    \draw[<->][red] (0, -1) .. controls (4, -1) and (3,-2.5) .. (4, -3) ;
                    \draw[<->][green!50!black] (0, -1.2) .. controls (1, -1.2) and (0,-3.2) .. (4, -3.2) ;
               \end{tikzpicture}
               %remove the next empty line and centering doesnt work. latex sucks.

               \end{centering}
               \caption{Two parts of the residual of $\alpha_1\times\alpha_2\times \alpha_3$ at $(x_1,x_2,x_3)$ that have comparable elements.
            }
            \label{fig-residual-fail-3}
        \end{figure}
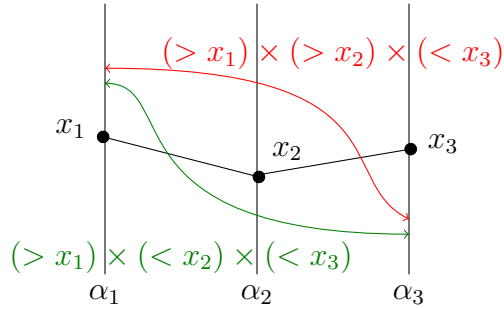
 
We need a method that will allow us to look further than the first element of the antichain. For this we extend the method of residuals with a game-theoretical point of view.

\subsection{Games}

\begin{definition}
$G_{X,\alpha}$ is a game for two players, let's call them Antoine (for antichain) and Odile (for ordinal), with the following rules:
\begin{itemize}
\item Each configuration of the game is a pair $(Y,\gamma)$ with $Y$ an antichain of $X$, and $\gamma\leq\alpha$.
\item The game begins in $(\emptyset,\alpha)$. Either Odile or Antoine begins.
\item At Odile's turn, she moves from configuration $(Y,\gamma)$ to $(Y,\gamma')$, with $\gamma'<\gamma$.
\item At Antoine's turn, he moves from configuration $(Y,\gamma)$ to $(Y \cup \{x\},\gamma)$, with $x \in X_{\perp Y}$.
\item The first player who cannot play loses.
\end{itemize} 
\end{definition}

This game is a specific case of the games defined in \cite{blass2008} and \cite{dzamonja2020}.
Since $X$ is FAC and $\alpha$ is WF, the players cannot play forever, so the game terminates.

\begin{lemma}[\cite{blass2008},\cite{dzamonja2020}]

$\w(X) \leq \alpha$ iff Odile has a winning strategy when Antoine begins. $\w(X) \geq \alpha$ iff
Antoine has a winning strategy when Odile begins.

\end{lemma}

Intuitively, one can see this game as playing along a branch of $\Inc(A)$: each time Antoine plays
$Y \leftarrow Y \cup \{x\}$, he moves from node $Y$ to its child $Y \cup \{x\}$. Odile has a winning strategy when she can play $\gamma \geq r(Y)$. Antoine has a winning strategy when he can play $Y$ such that $\gamma \leq r(Y)$.

To prove $\w(X) = \alpha$, we only need to exhibit two winning strategies, one for each player depending on who begins.

%Assume that $\w(A) \leq \alpha$, and Antoine begins. For any antichain $Y \neq \emptyset$, let $r(Y)$ be the rank of $Y$ in $\Inc (A)$, and $r(\emptyset)=\w(A)$. At the beginning of the game, we have $r(\emptyset) \leq \alpha$. 
%\emph{Invariant :} At the beginning of Antoine's turn, $r(Y) \leq \gamma$.
%Assume the invariant is true at configuration $(Y,\gamma)$. Antoine plays $Y\leftarrow Y\cup\{x\}$. Since
%$r(Y \cup \{x\})<r(Y) \leq \gamma$, there exists $\gamma'<\gamma$ such that $r(
%Y \cup \{x\})\leq \gamma'$. Hence Odile always has a move that preserves the invariant, therefore she has a winning strategy.
%
%Now assume that $\w(A) \geq \alpha$, and Odile begins. Thus at the beginning of the game, $r(\emptyset) \geq \gamma$.
%\emph{Invariant :} At the beginning of Odile's turn, $r(Y) \geq \gamma$. 
%Assume the invariant is true at configuration $(Y,\gamma)$. Odile plays $\gamma \leftarrow \gamma'$. Since
%$r(Y \cup \{x\})<r(Y) \leq \gamma\leq r(Y)$, and $r(Y)= \sup \bigl\{r(Y\cup\{x\}) +1 | Y\cup\{x\} \text{ is an antichain}\bigr\}$, there exists $x \in A$ such that $Y\cup\{x\}$ is an antichain and $r(Y\cup^\{x\})\geq \gamma'$. Hence Antoine always has a move that preserves the invariant, therefore he has a winning strategy.
%
%If one player has a winning strategy when they begin, then they also has a winning strategy when the other player begins, thus by contraposition:

\section{How to prove lower bounds}
\label{section-lowerbound}

\subsection{Combining strategies}

Here we introduce a method to combine  several winning strategies for Antoine in order to prove lower bounds on the width of complex wqos.

We denote $A\perp B$ with $A$ and $B$ two subsets of a wqo when for any $a\in A, b\in B$, $a\perp b$ in this wqo.
We say $A_1,\dots,A_m$ is an \emph{incomparable} family of subsets of $A$ when $A_i \perp A_j$ for any $i\neq j$. Observe that, for an incomparable family, we have $A \geqstruct \bigsqcup_i A_i$ thus $\w(A) \geq \bigoplus_i \w(A_i)$. However, we can do almost as well with a weaker condition on the $A_i$s.

We say $A_1,\dots,A_m$ is a \emph{quasi-incomparable} family of subsets of $A$ if for any $i \in [1,m]$, for every finite set $Y \subseteq A_1\cup \dots\cup A_{i-1}$, there exists $A_i' \subseteq A_i$ isomorphic to $A_i$ such that $A_i' \perp Y$. Note that this notion is sensitive to the way we number the $A_i$s.

\begin{lemma}[How to combine winning strategies for Antoine]
\label{lem-Antoine-strat}
\quad\newline Let $A_1,\dots,A_m$ be a quasi-incomparable family of subsets of $A$. Then $\w(A)\geq \w(A_m) + \dots + \w(A_1)$.
\end{lemma}

\begin{proof}

Let us note $\alpha_i$ for $\w(A_i)$

For any $i\in[1,m]$, Antoine has a winning strategy $S_i$ on $G_{A_i,\alpha_i}$ when Odile begins.
We want to combine those $S_i$ into a winning strategy for Antoine on
$G_{A,\alpha_m+\dots+\alpha_1}$ when Odile begins.

Intuitively, the game is played in $m$ phases. Odile goes through the sum $\alpha_m+\dots+\alpha_1$ from right to left, which means that at the $j$-th phase she is decreasing the term $\alpha_j$, while Antoine plays his strategies $S_j$ on $A_j'$ the subset of $A_j$ incomparable to the antichain built by Antoine during phases $1$ to $j-1$. As the subsets $A'_i$ are only isomorphic to the $A_i$s on which the strategies are defined, we have to remember to shift the $S_i$s.

More formally, assume that Odile has selected some ordinal $\gamma$, and Antoine has selected an antichain $Y\in A_1 \cup \dots \cup A_k$.
Now $\gamma$ can be written in a unique way as $\gamma = \alpha_m + \dots + \alpha_{k+1} + \sigma$ with $\sigma< \alpha_k$ for some $k\in[1,m]$: the game is in phase $k$.  

\emph{Our invariant:} During phase $k$, $Y\in A_1 \cup \dots \cup A_k$. By definition of a quasi-incomparable family, there exists $A_k'\subseteq A_k$ isomorphic to $A_k$ such that $A_k' \perp (Y\setminus A_k)$. The antichain $Y\cap A_k$ comes from the strategy $S_k$ played on $A_k'$.

Now it is Odile's turn, and she selects some $\gamma' <\gamma$. We know $\gamma'$ is either above or strictly below $\alpha_m + \dots + \alpha_{k+1}$:
\begin{itemize}

\item If $\gamma = \alpha_m + \dots + \alpha_{j+1} + \sigma'$ then we move to the $j$th phase of the game: By definition of a quasi-incomparable family, there exists $A_j'\subseteq A_j$ isomorphic to $A_j$ such that $A_j' \perp Y$. We follow the strategy $S_j$ on $A_j'$ which selects some $x\in A_j'$. Since $A_j' \perp Y$, $Y\leftarrow Y \cup \{x\}$ is still an antichain.

\item If $\gamma = \alpha_m + \dots + \alpha_{k+1} + \sigma'$ for some $\sigma' <\sigma$, then we can keep applying the strategy $S_k$ on $A_k'$, which selects some $x \perp Y \cap A_k'$. Since $A_k' \perp (Y\setminus A_k)$, $Y \leftarrow Y \cup \{x\}$ is still an antichain.

\end{itemize}

\end{proof}

\subsection{Lower bound for self-residual wqos}

\begin{definition}[Self-residual]
Let $A$ be a wqo. $A$ is \emph{self-residual} if for any $x\in A$, $A_{\not \leq x}$ contains an isomorphic copy of $A$.
\end{definition}
Let us illustrate this notion with an example: An ordinal $\alpha$ is said to be \emph{indecomposable}, or \emph{additive principal}, if for any $\beta,\gamma <\alpha$, $\beta + \gamma < \alpha$. All indecomposable ordinals are of the form $\alpha = \om^{\alpha'}$.

Then for all infinite indecomposable ordinal $\alpha$, for any $x<\alpha$, $\alpha_{\not\leq x} \equiv \alpha - (x + 1) = \alpha$
, because $x+1<\alpha$ and by definition of indecomposable. Therefore $\alpha$ is self-residual .
 
 \medskip
The notion of self-residual is compatible with cartesian product: if $A$ and $B$ are self-residual wqos, then $A\times B$ is self-residual. In particular, a cartesian product of $n$ infinite indecomposable ordinal is self-residual.

If $A$ is self-residual, then for all finite $Y\subseteq A$, $A_{\not\leq Y}$ contains an isomorphic copy of $A$ (by induction on the size of $Y$).

Here is an application of Lemma \ref{lem-Antoine-strat} that will be useful in Section \ref{section-indecomposable}:

\begin{lemma}
\label{lemma-cdot-k}
Let $A,B$ be two wqos such that $A$ is self-residual, and $k\in\mathbb{N}$. Then $\w(A\times (B \cdot k)) \geq \w(A\times B) \cdot k$.
\end{lemma}

\begin{proof}

Let $B_1,\dots,B_k$ be disjoint copies of $B$. Then $B\cdot k$ is isomorphic to the lexicographic sum $B_k + \dots + B_1$. We claim that $(A\times B_i)_{i\in[1,n]}$ is a quasi-incomparable family of subsets of 
$A \times (B\cdot k)$:

Fix $j \in [1,k]$ and $Y\subset (A\times B_1) \cup \dots \cup (A\times B_{j-1})$ finite (Figure \ref{fig-cdot-k} illustrates the case $j=3$).
Let $Y_{|A} \eqdef \{a\in A \;|\;(a,b)\in Y, b \in B \cdot k \}$. 
We want to find a subset of $C\subseteq A\times B_j$ isomorphic to $A\times B_j$ such that $C\perp Y$.
Since $A$ is self-residual, $A_{\not\leq Y_{|A} }$ contains an isomorphic copy of $A$, hence $ A_{\not\leq Y_{|A} } \times B_j \equiv A \times B_j$. Since for any $(a,b) \in A_{\not\leq Y_{|A} } \times B_j$,$(a',b')\in Y$, $a >_A a'$ and $b<_{B\cdot k} b'$ so $(a,b)\perp (a',b')$.

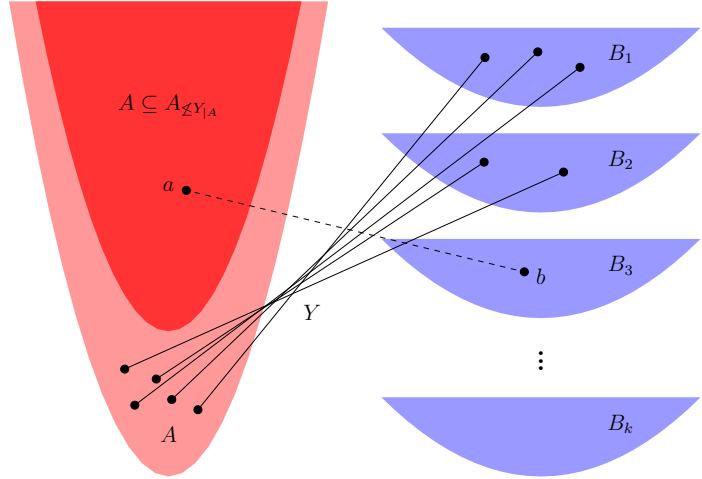
\begin{figure}[h]

               \centering
               \scalebox{0.7}{
                \begin{tikzpicture}
                \fill [color=red!40] (-3,9) -- (3,9)--plot[domain=3:-3](\x,\x*\x)--cycle;
                \fill [color=red!80] (-2.5,9) -- (2.5,9)--plot[domain=2.5:-2.5](\x,\x*\x+2.75)--cycle;
                \fill [color=blue!40](4,1.5)--(10,1.5)--plot[domain=3:-3](\x+7,\x*\x/6)--cycle;
                \fill [color=blue!40](4,4.5)--(10,4.5)--plot[domain=3:-3](\x+7,\x*\x/6+3)--cycle;
                \fill [color=blue!40](4,6.5)--(10,6.5)--plot[domain=3:-3](\x+7,\x*\x/6+5)--cycle;
                \fill [color=blue!40](4,8.5)--(10,8.5)--plot[domain=3:-3](\x+7,\x*\x/6+7)--cycle;
                \draw [*-*] (.5,1.2) -- (6,8);
                \draw [*-*] (0,1.4) -- (7,8.1);
                \draw [*-*] (-.3,1.8) -- (6,6);
                \draw [*-*] (-.9,2) -- (7.5,5.8);
                \draw [*-*] (-.7,1.3) -- (7.8,7.8);
                \node (Ares) at (0,7){$A \subseteq A_{\not\leq Y_{|A}}$};
                \node (A) at (0,0.8){$A$};
                \node (Bun) at (8.5,8){$B_1$};
                \node (Bd) at (8.5,6){$B_2$};
                \node (Bt) at (8.5,4){$B_3$};
                \node (Bk) at (8.5,1){$B_k$};
                \node (Y) at (2.7,3.1){$Y$};
                \node (dots) at (7,2.3){\Huge $\vdots$};
                \node (a) at (0,5.5){$a$};
                \node (b) at (7,3.8){$b$};
                \draw [*-*][dashed](a) -- (b);
                
               \end{tikzpicture}
               
}
               \caption{All elements $(a,b)$ of $A_{\not\leq Y_{|A}}\times(B_k + \dots + B_3)$ are incomparable to $Y\subseteq A \times (B_2 + B_1)$.}
               \label{fig-cdot-k}
        \end{figure}

Therefore $(A\times B_i)_{i\in[1,n]}$ is a quasi-incomparable family, so according to Lemma \ref{lem-shift-iso}, 
$$\w\left(A\times \left(\sum_{i=k}^1 B_i\right)\right) \geq \sum_{i=k}^1\w(A\times B_i) = \w(A\times B) \cdot k\;.\qedhere$$

\end{proof}

When $\w(A\times B)$ is a finite multiple of an indecomposable, we have the other side of the equality: $A\times (B\cdot k) \geqaug A \times B \times \Gamma_k$, so $\w(A\times (B\cdot k)) \leq \w(A\times B) \otimes k = \w(A\times B) \cdot k$, therefore $\w(A\times (B\cdot k)) = \w(A\times B) \cdot k$.

\subsection{Lower bound for transferable wqos}

A wqo $A$ is \emph{transferable} if $\w(A_{\not\leq Y}) = \w(A)$ for any finite $Y\in A$. Observe that it is a weaker condition than self-residuality:

\begin{lemma}
A self-residual wqo is transferable.
\end{lemma}
\begin{proof}
Let $A$ be a self-residual wqo, i.e., for any $x \in A$, $A_{\not\leq x}$ contains an isomorphic copy of $A$. Thus by induction on the size of $Y$, $A_{\not\leq Y}$ contains an isomorphic copy of $A$. Therefore $\w(A_{\not\leq Y}) = \w(A)$, so $A$ is transferable.
\end{proof}

Transferability is used in \cite{dzamonja2020} to prove a lower bound:

\begin{lemma}[Theorem 4.16 of \cite{dzamonja2020}]
\label{weak-lbiftransferable}
Suppose that $A$ is a transferable wqo and $\beta$ is an ordinal.
Then $\w(A \times \beta) \geq \w(A) \cdot \beta$.
\end{lemma}

From this lemma, we can deduce a more general version of itself which we will use in Section \ref{section-indecomposable}.

\begin{lemma}
\label{lbiftransferable}
Suppose that $A$ is a transferable wqo and $B$ any wqo. Then $\w(A\times B)\geq \w(A)\cdot\o(B)$.
\end{lemma}
\begin{proof}
Observe that $B \leqaug \o(B)$. Therefore $\w(A\times B) \geq \w(A \times \o(B)) \geq \w(A) \cdot \o(B)$.
\end{proof}

Combined with the method of residuals, Lemma \ref{lbiftransferable}  allows us to compute the width of simple examples:

We note $A^{\times n}\eqdef A \times \dots \times A$ the cartesian product of $n$ copies of a wqo $A$.

\begin{proposition}
\label{prop-om-times-n}

$\w(\om^{\times n}) = \om^{n-1}$ for $n\geq 1$.
\end{proposition}
\begin{proof}
Case $n =1$: $\w(\om)=1$. 

If $n>1$,
$
\w(\om^{\times n})\geq \w(\om) \cdot \o(\om^{\times (n-1)}) = \om^{n-1}$ according to Lemma \ref{lbiftransferable}.

Let us prove the upper bound by induction on $n$, initialized in $n=1$: Assume $w(\omega^{\times n})=\omega^{n-1}$ for some $n$. Let $m=(m_0,\dots,m_n)$ be any element of $\omega^{\times (n + 1)}$, $m'=(m_1,\dots,m_n)$, and $k$ the cardinal of $(\omega^{\times n})_{< m'}$.
Then:

\begin{align*}
(\omega^{\times (n + 1)})_{\perp m} & \geqaug (< m_0)\times(\omega^{\times n})_{>m'} \;\sqcup\; (> m_0)\times(\omega^{\times n})_{<m'}\;\sqcup\; \{m_0\}\times (\omega^{\times n})_{\perp m'} \\
& \geqaug \Gamma_{m_0} \times \omega^{\times n} \;\sqcup\; \omega \times \Gamma_k \;\sqcup\;  (\omega^{\times n})_{\perp m'} \; .
\end{align*}

Therefore by induction hypothesis $\w\left(\omega^{\times (n + 1)}_{\perp m}\right)\leq \omega^{n-1}\cdot m_0 \;\oplus\; k \;\oplus\; \gamma $ with $\gamma < \omega^{n-1}$. 

Thus by the method of residuals: $\w\left(\omega^{\times (n + 1)}\right)=\underset{m}{\sup}\left\{\w(\omega^{\times (n + 1)}_{\perp m}) + 1 \right\} \leq \omega^{n}$.

\end{proof}

\begin{proposition}
\label{prop-omom-times-n}
$\w((\omom)^{\times n}) = \om^{\om\cdot n}$ for $n\geq 2$.
\end{proposition}
\begin{proof}
The case $n=2$ is an application of Theorem \ref{thm-abraham-limit}. If $n>2$:
\begin{align*}
\w((\omom)^{\times n})&\leq \o((\omom)^{\times n}) = \om^{\om\cdot n}
\text { according to Lemma \ref{lem-w-leq-o},}\\
\w((\omom)^{\times n})&\geq \w(\omom\times\omom) \cdot \o((\omom)^{\times (n-2)}) = \om^{\om\cdot 2}\cdot \om^{\om\cdot (n-2)}= \om^{\om\cdot n}
\end{align*}
according to Lemmas \ref{lem-product-o} and \ref{lbiftransferable}.
\end{proof}

\section{Product of indecomposable ordinals}
\label{section-indecomposable}

We want to compute $\w(\om^{\alpha_1}\times\dots\times \om^{\alpha_n})$ for $\alpha_1,\dots,\alpha_n > 0$. We recall this result from \cite{abraham87} for $n=2$:

\begin{lemma}
\label{lem-eta-2}
Let $\alpha_1,\alpha_2 > 0$. Then $\w(\om^{\alpha_1}\times \om^{\alpha_2})= \om^{\eta}$ with $\eta ={1 +\bigl((\alpha_1-1) \oplus (\alpha_2-1)\bigr)}$.

\end{lemma}

Observe that for any ordinal $\alpha>0$, $\alpha - 1 = \alpha$ iff $\alpha$ is infinite. 
Thus $\eta$ can be simplified depending on the finiteness or infiniteness of $\alpha_1$ and $\alpha_2$:
\begin{itemize}
\item If $\alpha_1,\alpha_2$ are finite then 
$\eta =\alpha_1 \nfp\alpha_2 -1$.
\item If $\alpha_1$ is infinite and $\alpha_2$ finite then 
$\eta = \alpha_1 \nfp \alpha_2 - 1$.
\item If $\alpha_1,\alpha_2$ are infinite then 
$\eta = \alpha_1 \oplus \alpha_2$.
\end{itemize}

We can simplify the statement of Lemma \ref{lem-eta-2} if we order $\alpha_1 \geq \alpha_2$ without loss of generality, which gives us $\w(\om^{\alpha_1}\times \om^{\alpha_2})= \om^{\alpha_1 \oplus (\alpha_2-1)}$.

\begin{theorem}
\label{thm-eta}
Let $X= \om^{\alpha_1}\times\dots\times \om^{\alpha_n}$, with $n \geq 2$ and   $\alpha_1 \geq\dots \geq \alpha_n>0$. 
Then $\w (X) = \om^{\eta}$ with $\eta =\alpha_1 \oplus \bigl( (\alpha_2 \oplus \dots \oplus \alpha_n) - 1 \bigr )$, i.e.:
\begin{itemize}
\item If $\alpha_1, \dots,\alpha_n$ are finite then 
$\eta = \alpha_1 \nfp \dots \nfp \alpha_n -1$.
\item If $\alpha_1$ is infinite and $\alpha_2\dots,\alpha_n$ are finite then 
$\eta = \alpha_1 \nfp (\alpha_2 \nfp \dots \nfp \alpha_n - 1)$.
\item If $\alpha_1, \dots,\alpha_k$ are infinite with $k \geq 2$ then 
$\eta = \alpha_1 \oplus \dots \oplus \alpha_n$.
\end{itemize}

\end{theorem}

\begin{proof}
Case $n=2$ is treated in \cite{abraham87}. Assume that $n>2$.

Let $k\in[0,n]$ be such that $\alpha_1,\dots,\alpha_k$ are infinite and $\alpha_{k+1},\dots,\alpha_n$ are finite. Inside both the proofs of the lower bound and upper bound, the cases $k=0$, $k=1$, $2\leq k < n$ and $k=n$ will be treated separately when necessary.

We prove the lower bound with Lemma \ref{lbiftransferable}:

For any ordinal $\alpha>0$, $\om^\alpha$ is principal additive so self-residual. Therefore
$\om^{\alpha_1}\times \dots \times \om^{\alpha_j}$ is self-residual for all $j\leq n$ hence transferable.

\begin{itemize}
\item If $k=0$, $\w (X) \geq \w(\om^{\alpha_1}\times \om^{\alpha_2}) \cdot \o(\om^{\alpha_3}\times\dots\times \om^{\alpha_n}) = \om^{(\alpha_1 \nfp \alpha_2 - 1) + (\alpha_3 \oplus \dots \oplus \alpha_n)} = \om^{\alpha_1 \nfp \dots \nfp \alpha_n -1}$.
\item If $k=1$, $\w (X) \geq \w(\om^{\alpha_1}\times \om^{\alpha_2}) \cdot \o(\om^{\alpha_3}\times\dots\times \om^{\alpha_n}) = \om^{(\alpha_1 \nfp (\alpha_2 - 1)) + (\alpha_3 \oplus \dots \oplus \alpha_n)} = \om^{\alpha_1 \nfp (\alpha_2 \nfp \dots \nfp \alpha_n -1)}$.
\item If $2 \leq k < n$, then $\w (X) \geq \w(\om^{\alpha_1}\times \dots \times \om^{\alpha_k}) \cdot \o(\om^{\alpha_{k+1}}\times\dots\times \om^{\alpha_n}) =
\om^{(\alpha_1 \oplus \dots \oplus \alpha_k) + (\alpha_{k+1} \oplus \dots \oplus \alpha_n)} = \om^{\alpha_1 \oplus \dots \oplus \alpha_n}$ by induction on $n$.
\item If $k=n$, then we cannot use Lemma \ref{lbiftransferable} (it would give a $+$ instead of a $\oplus$) so we proceed by induction on $(\alpha_1,\dots,\alpha_n)$ with the cartesian product ordering: 

We already know that $\w((\omom)^{\times n}) = \om^{\om\cdot n}$ for any $n\geq 2$ from Proposition \ref{prop-omom-times-n}.
Let $i$ such that $\alpha_i$ has the smallest last exponent $\rho$, i.e., $\alpha_i = \alpha_i' \nfp \om^{\rho}$ 
and $\alpha_1 \oplus \dots \oplus \alpha_n = (\alpha_1 \oplus \dots \oplus \alpha_i' \oplus \dots \oplus \alpha_n) \nfp \om^\rho$. The ordinal $\rho$ can either be null, a successor or a limit ordinal:
 \begin{itemize}
 \item If $\rho = 0$ then $X \geqstruct \om^{\alpha_1}\times\dots\times \om^{\alpha_i'} \cdot k \times \dots \times \om^{\alpha_n}$ for any $k<\om$, and $\alpha_i'$ is infinite. 
 Since  $\om^{\alpha_1}\times\dots\times \om^{\alpha_{i-1}}\times  \om^{\alpha_{i-+}} \times \dots \times \om^{\alpha_n}$ is self-residual, then according to Lemma \ref{lemma-cdot-k}, 
 $$\w(X) \geq \w(\om^{\alpha_1}\times\dots\times \om^{\alpha_i'} \times \dots \times \om^{\alpha_n})\cdot k
 =\om^{\alpha_1 \oplus \dots \oplus \alpha_i' \oplus \dots \oplus \alpha_n}\cdot k$$ by induction hypothesis, so $\w(X) \geq \om^{\alpha_1 \oplus \dots \oplus \alpha_n}$.
 \item Otherwise if $\rho = \rho' + 1$ is a successor then $X \geqstruct \om^{\alpha_1}\times\dots\times \om^{\alpha_i' \nfp \om^{\rho'} \cdot k} \times \dots \times \om^{\alpha_n}$. Now $\alpha_i'\, \nfp\, \om^{\rho'} \cdot k$ is either infinite or finite (case where $\alpha_i' = 0$ and $\rho = 1$). In the first case we call on our induction hypothesis, in the second case we are in the situation where more than two exponents are infinite but not all of them, which we already treated. In both cases, we get:
 $$\w(X) \geq \om^{\alpha_1 \oplus \dots \oplus (\alpha_i' \nfp \om^{\rho'} \cdot k) \oplus \dots \oplus \alpha_n}
 =\om^{\alpha_1 \oplus \dots \oplus \alpha_i' \oplus \dots \oplus \alpha_n}\cdot \om^{\om^{\rho'} \cdot k}$$ so $\w(X) \geq \om^{\alpha_1 \oplus \dots \oplus \alpha_n}$.
 \item If $\rho = \sup_i \;\rho_i$ is a limit, then we reason as in the successor case, with $\alpha_i' \nfp \om^{\rho_i}$ instead of $\alpha_i' \nfp \om^{\rho'} \cdot k$

\end{itemize}

\end{itemize}

Now for the upper bound: 
In the case $k\geq 2$, we
already know from Lemma \ref{lem-w-leq-o} that $\w(X) \leq \o(X)=\om^{\alpha_1\oplus\dots\oplus \alpha_n}$.

When $k\leq 1$ we prove the upper bound by induction on $(\alpha_1,\dots,\alpha_n)$, using the methods of residuals:

We already know from Proposition \ref{prop-om-times-n} that $\w(\om^{\times n}) = \om^{n - 1}$. 

Let $x\in X$. For all $i\in[1,n]$ there exists $0\leq\alpha_i' < \alpha_i$ and $m_i \in\mathbb{N}$ such that $x_i\leq \om^{\alpha_i'}\cdot m_i<\om^{\alpha_i}$. The residual $X_{\perp x}$ is included in an augmentation of a disjoint sum of terms of the form
$\underset{i\in I}{\bm{\times}} (<x_i) \times \underset{i\not\in I}{\bm{\times}} (\geq x_i)$ with $I \subsetneq [1,n], I \neq \emptyset$. Hence:

\begin{align*}
\w(X_{\perp x}) &\leq \bigoplus_I \w\left(\underset{i\in I}{\bm{\times}} (<x_i) \times \underset{i\not\in I}{\bm{\times}} (\geq x_i)\right)\\
&\leq \bigoplus_I \w\left(\underset{i\in I}{\bm{\times}} (\om^{\alpha_i'}\times \Gamma_{m_i}) \times \underset{i\not\in I}{\bm{\times}} \om^{\alpha_i}\right)\\
&\leq \bigoplus_I \w\left(\underset{i\in I}{\bm{\times}} \om^{\alpha_i'} \times \underset{i\not\in I}{\bm{\times}} \om^{\alpha_i}\right)\cdot \prod_{i\in I} m_i\\
&\leq \left(\bigoplus_I \w\left(\underset{i\in I}{\bm{\times}} \om^{\alpha_i'} \times \underset{i\not\in I}{\bm{\times}} \om^{\alpha_i}\right)\right)\cdot m\;.
\end{align*}

By induction hypothesis, 
 $\w\left(\underset{i\in I}{\bm{\times}} \om^{\alpha_i'} \times \underset{i\not\in I}{\bm{\times}} \om^{\alpha_i}\right) = \om^{\eta'}$ for some $\eta' < \alpha_1 \nfp (\alpha_2 \nfp \dots \nfp \alpha_n -1)$.
Therefore $\w(X) \leq \om^{\alpha_1 \oplus (\alpha_2 + \dots + \alpha_n -1)}$.

\end{proof}

Theorem \ref{thm-eta} can be extended to all ordinals $\alpha_i\geq 0$:
\begin{theorem}
\label{thm-eta-Z}
Let $X= \om^{\alpha_1}\times\dots\times \om^{\alpha_n}$, with $n \geq 1$ and $\alpha_1\geq\dots\alpha_n \geq 0$ be $n$ ordinals. If $\alpha_2 = \dots = \alpha_n = 0$ then $\w(X)=1$, otherwise
$\w(X) = \om^{\alpha_1 \oplus ((\alpha_2 \oplus \dots \oplus \alpha_n) - 1)}$.
\end{theorem}
\begin{proof}
Let $k \leq n$ such that $\alpha_1\geq\dots\geq\alpha_k > 0 = \alpha_{k+1} =\dots = \alpha_n$.
If $k=0$ or $1$, then $X\equiv\om^{\alpha_1}$ so $\w(X)=1$. Otherwise $k\geq 2$, and $X \equiv \om^{\alpha_1}\times\dots\times\om^{\alpha_k}$ so according to Theorem \ref{thm-eta}, $\w(X)  = \om^{\alpha_1 \oplus ((\alpha_2 \oplus \dots \oplus \alpha_k) - 1)} = \om^{\alpha_1 \oplus ((\alpha_2 \oplus \dots \oplus \alpha_n) - 1)}$.
\end{proof}

An immediate corollary of Theorem \ref{thm-eta} which will be useful later is:
\begin{corollary}[Monotonicity]
\label{thm-eta-monotonicity}
Let $\alpha_1,\dots,\alpha_i\dots,\alpha_n>0$ be $n$ ordinals, and let $\alpha_i' > \alpha_i$ for some $i \leq n$. Then
$$\w(\om^{\alpha_1} \times \dots \times \om^{\alpha_n}) <
\w( \om^{\alpha_1} \times \dots \times \om^{\alpha_i'}\times \dots \times \om^{\alpha_n})\;.$$
\end{corollary}

\begin{proof}
The function $(\alpha_1,\dots,\alpha_n) \mapsto \alpha_1 \oplus \bigl( (\alpha_2 \oplus \dots \oplus \alpha_n) - 1\bigr)$ is strictly increasing, because the natural sum is strictly increasing, and the left subtraction (the only kind of subtraction defined for ordinals) is also strictly increasing in its left argument.
\end{proof}

\section{Product of infinite ordinals}
\label{section-infinite}

In this section we extend the width of the product of indecomposable ordinals (Theorem \ref{thm-eta-Z}) to the width of the product of infinite ordinals (Theorem \ref{limit-formula}).

Let $X = \alpha_1 \times \dots \times \alpha_n$ be a cartesian product of $n$ infinite ordinals. 
Each $\alpha_i$ is written in Cantor normal form without multiplicities as $\alpha_i = \sum_{j<l_i} \om^{\alpha_{i,j}}$, with $\alpha_{i,0}\geq\dots\geq \alpha_{i,l_i-1}$ for $i \in [1,n]$.

We partition $X$ into \emph{slices}: for any $s=(s(1),\dots,s(n))\in Sl(X)\eqdef l_1\times\dots\times l_n$, we define the slice $X_s$ as 
\begin{align*}
X_s &\eqdef \underset{i\in [1,n]}{\bm{\times}} X_{s,i}\\
\text{with } X_{s,i} &\eqdef
\left\{\delta \;\middle|\;\sum_{j< s(i)} \om^{\alpha_{i,j}} 
\leq \delta <
\sum_{j\leq s(i)} \om^{\alpha_{i,j}}\right\} \;.
\end{align*}

By abuse of language, we also call $s$ a slice.

Observe that as a substructure of $X$, $X_s$ is isomorphic to $\underset{i\in [1,n]}{\bm{\times}} \om^{\alpha_{i,s_i}}$. Therefore we know how to compute $\w(X_s)$ thanks to Theorem \ref{thm-eta-Z}.

We say $s\in Sl(X)$ is \emph{grounded} if there exists $k \in [1,n]$ such that $s(k)=0$. Let $Gr(X)\eqdef \{s\in Sl(X) \;|\;\exists k\in[1,n], s(k)=0\}$ the set of grounded slices.

 \begin{figure}[h]
            
               \begin{centering}
                \begin{tikzpicture}
                	\node (A) at (0, -.5){$\omom + \om$};
                	\node (B) at (2, -.5){$\om\cdot 3$};
                	\node (A) at (4, -.5){$\om^3 + \om^2 + 1$};
                	\node [red](Xs) at (4.8, 4.8) {$X_s$};
                	\node [green!50!black](Xt) at (4.8, 1.5) {$X_t$};
                    \draw [thick](0,0) -- (0,3) node  [midway, left]{$\omom$};
                    \node[draw,circle,inner sep=1pt,fill] at (0,0) {};
                    \draw [dotted, very thick](0,3) -- (0,3.4);
                    \node[draw,circle,inner sep=1pt,fill] at (0,3.6) {};
                    \draw [thick](0,3.6) -- (0,4.5)node  [midway, left]{$\om$};
                    \draw [dotted, very thick](0,4.5) -- (0,4.9);
                    \node[draw,circle,inner sep=1pt,fill] at (2,0) {};
                    \draw [thick](2,0) -- (2,1)node  [midway, left]{$\om$};
                    \draw [dotted, very thick](2,1) -- (2,1.4);
                    \node[draw,circle,inner sep=1pt,fill] at (2,1.6) {};
                    \draw [thick](2,1.6) -- (2,2.5)node  [midway, left]{$\om$};
                    \draw [dotted, very thick](2,2.5) -- (2,2.9);
                    \node[draw,circle,inner sep=1pt,fill] at (2,3.1) {};
                    \draw [thick](2,3.1) -- (2,4)node  [midway, left]{$\om$};
                    \draw [dotted, very thick](2,4) -- (2,4.4);
                    \node[draw,circle,inner sep=1pt,fill] at (4,0) {};
                    \draw [thick](4,0) -- (4,2)node  [midway, left]{$\om^3$};
                    \draw [dotted, very thick](4,2) -- (4,2.4);
                    \node[draw,circle,inner sep=1pt,fill] at (4,2.6) {};
                    \draw [thick](4,2.6) -- (4,4)node  [midway, left]{$\om^2$};
                    \draw [dotted, very thick](4,4) -- (4,4.4);
                    \node[draw,circle,inner sep=1pt,fill] at (4,4.7) {};
                    \draw [red] plot[smooth]coordinates{(-.1, 3.3)(-.3,3.7)(-.3,5)(.2,4.9)(2,1.4)(3.9,4.8)(4.2,4.65)(3.9,4.2)(2,-.1)(-.1,3.3)};
                    \draw [green!50!black] plot[smooth]coordinates{(0, 3.3)(-.2,3.7)(-.2,5)(.2,4.9)(2,4.5)(3.9,2.6)(4.2,2.1)(4,-.1)(2,2.9)(0,3.3)};
               \end{tikzpicture}
               %remove the next empty line and centering doesnt work. latex sucks.

               \end{centering}
               \caption{Slices and grounded slices: $X=(\omom + \om)\times(\om\cdot 3)\times (\om^3+\om^2+1)$,
               $X_s$ and $X_t$ for  $s = (1,0,2)$, $t=(1,2,0)$ are both grounded.
            }
        \end{figure}
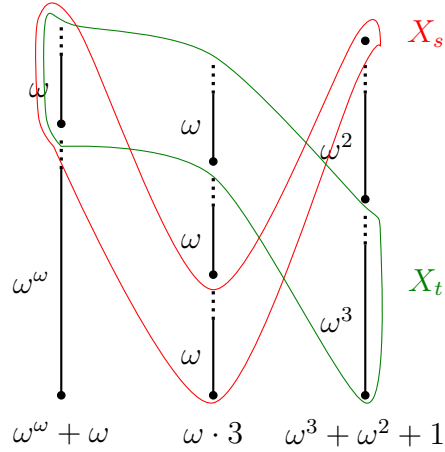

We compute $\w(X)$ by slices:

\begin{theorem}
\label{limit-formula}
\begin{equation*}
\w(X) = 
	\bigoplus_{s \in Gr(X)} 
	\w(X_s)\; .
	\end{equation*}
\end{theorem}

First we prove the upper bound $\w(X) \leq
	\bigoplus_{s \in Gr(X)} 
	\w(X_s)$.

\begin{proof}[Proof of the upper bound of Theorem \ref{limit-formula}]
\quad\newline

For any two $s,t \in Sl(X)$, we write $s \prec t$ iff for all $i \in [1,n]$, $s(i) < t(i)$. Observe that, for any $s \prec t$:
\begin{itemize}
\item for any $x\in X_s,\; x' \in X_t$, we have $x <_X x'$.
\item for any $i \in [1,n]$, $\alpha_{i,s(i)}\geq \alpha_{i,t(i)}$, so $X_s \geqstruct X_t$, therefore $\w(X_s) \geq \w(X_t)$.
\end{itemize} 

We define a surjective function $g:Sl(X)\rightarrow Gr(X)$ which associates with any slice a grounded slice:
$$g(s)(i) \eqdef s(i) - k 
\text{ with } k = \min_{i\in[1,n]}  s(i)\;.$$

This surjection has interesting properties: \begin{itemize}
\item If $s$ is grounded then $g(s)=s$, otherwise $g(s)\prec s$. Thus $\w(X_s)\leq \w(X_{g(s)})$.
\item For any $s\neq t \in Sl(X)$ such that $g(s) = g(t)$, $ s \prec t \text{ or } s \succ t$.
Therefore for any $s\in Gr(X)$, the union of all the $X_{s'}$ such that $g(s')=s$ is a lexicographic sum in $X$.
\end{itemize}
 
Thus $X$ can be expressed as an augmentation of the disjoint sum of lexicographic sums of $X_s$ gathered by the image of $s$ through $g$:
$$X \geqaug \underset{s \in Gr(X)}{\bigsqcup}\;\underset{s'\in g^{-1}(s)}{\sum} X_{s'}\;,$$

which implies, according to Lemmas \ref{lem-DSS-sqcup} and \ref{lem-DSS-sum},
$$\w(X) \leq \underset{s \in Gr(X)}{\bigoplus}\;\underset{s'\in g^{-1}(s)}{\max} \w(X_{s'})
= \underset{s \in Gr(X)}{\bigoplus} \w(X_s)\;. $$

\end{proof}

We need to introduce a few notations before proving the lower bound of Theorem \ref{limit-formula}.

Let $Y$ be a finite set of elements of $X$.
We define a function $\shift$ which outputs a subset of $X$:
$\shift(X_s,Y)\eqdef \underset{i\in [1,n]}{\bm{\times}} \shift(X_{s,i},Y)$
with $$\shift(X_{s,i},Y)\eqdef \begin{cases}

\bigr{\{\xi(Y,i) 
< \delta <
 \om^{\alpha_{i,0}}\bigl\} \text{ if $s(i) = 0$,}}\\
{X_{s,i} \text{ otherwise,}}\\
\end{cases}$$
where we define $\xi(Y,i)$ as the max of the $i$-th components of elements of $Y$ which are less than $\om^{\alpha_i,0}$:
$$\xi(Y,i) \eqdef \max \;
\left\{y(i)\;|\;y \in Y, y(i) < \om^{\alpha_i,0}\right\}\;.
$$

If $\{y(i)\;|\;y \in Y, y(i) < \om^{\alpha_i,0}\} = \emptyset$ then let $\xi(Y,i) = \texttt{-}1$.
\begin{figure}[h]
\begin{centering}
                \begin{tikzpicture}
                	\node (A) at (0, -.5){$\omom + \om$};
                	\node (B) at (2, -.5){$\om\cdot 3$};
                	\node (C) at (4, -.5){$\om^3 + \om^2 + 1$};
                	\node [blue!70!black](Y) at (3,3){$Y$};
                	\node [green!70!black](Xt) at (5.4, 1.5) {$\shift(X_s,Y)$};
                    \draw [thick](0,0) -- (0,3);
                    \draw [dotted, very thick](0,3) -- (0,3.4);
                    \draw [thick](0,3.5) -- (0,4.5);
                    \draw [dotted, very thick](0,4.5) -- (0,4.9);
                    \draw [thick](2,0) -- (2,1);
                    \draw [dotted, very thick](2,1) -- (2,1.4);
                    \draw [thick](2,1.5) -- (2,2.5);
                    \draw [dotted, very thick](2,2.5) -- (2,2.9);
                    \draw [thick](2,3) -- (2,4);
                    \draw [dotted, very thick](2,4) -- (2,4.4);
                    \draw [thick](4,0) -- (4,2);
                    \draw [dotted, very thick](4,2) -- (4,2.4);
                    \draw [thick](4,2.5) -- (4,4);
                    \draw [dotted, very thick](4,4) -- (4,4.4);
                    \draw [ultra thin, *-] (4,4.6) -- (4,4.7);
                    \draw [green!70!black] plot[smooth]coordinates{(0, 3.3)(-.2,3.7)(-.2,5)(.2,4.9)(2,2.9)(3.9,2.6)(4.2,2.1)(4,0.8)(2,1.4)(0,3.3)};
                    \draw [blue!70!black](0,4.2) -- (2,3.8) -- (4,0.1);
                    \draw [blue!70!black](0,4.2) -- (2,3.6) -- (4,0.4);
                    \draw [blue!70!black](0,3.8) -- (2,.1) -- (4,4.65);
                    \draw [blue!70!black](0,4) -- (2,.5) -- (4,3.5);
                    \draw [blue!70!black](0,4) -- (2,.7) -- (4,3.3);
                    \node [red](tresh) at (4.9,0.7){$\xi(Y,3)$};
                    \draw [red, very thick](3.9,0.7) -- (4.1,0.7);
               \end{tikzpicture}
               %remove the next empty line and centering doesnt work. latex sucks.

               \end{centering}
               \caption{$\shift(X_s,Y)$, for $s=(1,1,0)$, is incomparable to $Y$}
        \end{figure}
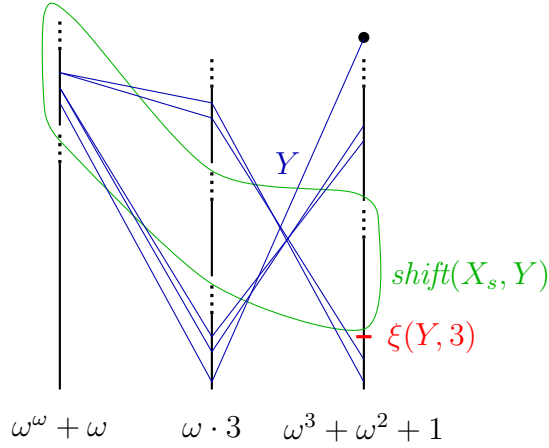

\begin{lemma}
\label{lem-shift-iso}
$\shift(X_{s})$ is isomorphic to $X_s$.
\end{lemma}
\begin{proof}
For any $i \in [1,n]$, $\shift(X_{s,i})$ is isomorphic to $X_{s,i}$:
If $s(i) > 0$ then $\shift(X_{s,i})=X_{s,i}$.
Otherwise $s(i) = 0$ and $\shift(X_{s,i})= \{\xi(Y,i) 
< \delta <
 \om^{\alpha_{i,0}}\} 
\approx  \om^{\alpha_{i,0}} - \xi(Y,i)$.
Since $\alpha_i$ is infinite, $\om^{\alpha_{i,0}}$ is infinite indecomposable. Therefore
$\om^{\alpha_{i,0}} - \xi(Y,i) \approx \om^{\alpha_{i,0}} = X_{s,i}$.
\end{proof}

\medskip

\begin{proof}[Proof of the lower bound of Theorem \ref{limit-formula}]
\quad\newline 
	
Proof idea: The notion of quasi-incomparable family is dependent on the order of the subsets. Thus we will order the grounded slices in such a way that $(X_s)_{s\in Gr(X)}$ is a quasi-incomparable family of subsets of $X$, and that Lemma \ref{lem-Antoine-strat} returns the expected result.
\medskip
	
	We first define the order.
	According to Theorem \ref{thm-eta-Z}, $\w(X_s)$ can be written under the form $\w(X_s)=\om^{\delta_s}$. Therefore we can order all the grounded slices as $s_1,\dots,s_L$ 
	(where $L = \prod l_i - \prod (l_i -1)$) in such a way that for any $i<j$, $\delta_{s_i}\leq \delta_{s_j}$. Then $\bigoplus_{s \in Gr(X)} 
	\w(X_s) = \w(X_{s_L}) \nfp \w(X_{s_{L-1}}) \nfp \dots \nfp \w(X_{s_1}) $.
	
	Since there are some slices $s\neq t$ such that $\delta_s =\delta_t$, we can refine our ordering of the slices. For any $i<j$, \begin{itemize}
	\item either $\delta_{s_i} < \delta_{s_j}$,
	\item or $\delta_{s_i} = \delta_{s_j}$ and 
	$\sum_{k\in[1,n]} s_i(k) \geq \sum_{k\in[1,n]} s_j(k)$.
	\end{itemize}
	From now on we write $\delta_i$ for $\delta_{s_i}$ and $X_i$ for $X_{s_i}$.
	\medskip
	
Now we will show that $X_{s_1},\dots, X_{s_L}$ form a quasi-incomparable family of subsets of $X$. 

Fix $k \in [1,L]$ and $Y \subseteq X_{s_1}\cup\dots\cup X_{s_{k-1}}$ a finite set. Then we define $X_{s_k}'\subseteq X_{s_k}$ as $\shift(X_{s_k},Y)$. According to Lemma \ref{lem-shift-iso}, $X_{s_k}'$ is indeed isomorphic to $X_{s_k}$. 

\medskip
Now we will show $X_{s_k}' \perp Y$. For any elements $y\in Y$, it is sufficient to find $i_1,i_2 \in [1,n]$ such that the $i_1$th component of $y$ is below $\shift(X_{s_k,i_1})$, and the $i_2$th component of $y$ is above $\shift(X_{s_k,i_2})$.

For any $j<k$, $s_j$ is grounded, hence there exists $i_1 \in [1,n]$ such that $s_j(i_1) = 0$. If $s_k(i_1) > 0$ then all elements $X_{s_j,i_1}$ are below 
$X_{s_k,i_1}$. Otherwise if $s_k(i_1) = 0$ then $\shift(X_{s_k,i_1}) = \{\xi(Y,i_1) 
< \delta <
 \om^{\alpha_{i_1,0}}\}$ and the $i_1$th component of all elements of $Y\cap X_{s_j}$ is below $\xi(Y,i_1)$.

We claim that there exists $i_2$ such that $s_j(i_2) > s_k(i_2)$, which means that all elements of $X_{s_j,i_2}$ are above 
$X_{s_k,i_2}$. 
Since $j<k$, we know $\delta_j \leq \delta_k$:
\begin{itemize}
\item If $\delta_j = \delta_k$ then $\sum_{i\in[1,n]} s_j(i) \geq \sum_{i\in[1,n]} s_k(i)$ so there exists $i_2$ such that $s_j(i_2) > s_k(i_2)$. 
\item Otherwise $\delta_j < \delta_k$. Assume for the sake of contradiction that for all $i \in [1,n]$, $s_j(i) \leq s_k(i)$. Then $\alpha_{i,s_j(i)} \geq \alpha_{i,s_k(i)}$. Thus Theorem \ref{thm-eta-monotonicity} leads us to a contradiction: $\delta_j \geq \delta_k$.
\end{itemize}

It follows that $X_{s_k}' \perp (Y \cap X_j)$ for any $j<k$.
We know $Y \subseteq X_{s_1}\cup\dots\cup X_{s_{k-1}}$ so $X_{s_k}' \perp Y$. Therefore $(X_{s_j})_{j \in [1,L]}$ is a quasi-incomparable family, hence according to Lemma \ref{lem-Antoine-strat}, $$\w(X) \;\geq\;\w(X_{s_L}) \nfp \dots \nfp \w(X_{s_1}) \;= \bigoplus_{s \in Gr(X)} 
	\w(X_s)\;.$$
	
\end{proof}
\medskip

\subsection{Alternative expressions for $w(X)$}

Our formula to compute $\w(X)$ is expressed for ordinals written in normal form without multiplicities. However, there are two other ways to write ordinals in normal form, which are more commonly used:

\begin{align*}
\alpha_i 
		&= \sum_{j<l_i} \om^{\alpha_{i,j}} &\text{ (this is the one we used until now),}\\
		&= \sum_{j<l'_i} \om^{\alpha'_{i,j}}\cdot a_{i,j} &\text{ in developed normal form,}\\
		&= \om^{\alpha_i'}\cdot a_i \nfp \sigma_i &\text{ in short normal form.} \\
\end{align*}

When we go from the normal form without multiplicities to the developed normal form, it allows us to regroup several slices $s\in Sl(X)$ into one slice $t \in Sl'(X) \eqdef l_1'\times\dots\times l_n'$. We denote $Gr'(X)$ as the grounded slices of $Sl'(X)$ We define $f$ as the function from $Sl(X)$ to $Sl'(X)$ such that $f(s)=t$ if $\alpha_{i,s(i)}=\alpha'_{i,t(i)}$ for any $i\in[1,n]$. Let $X_t \eqdef \underset{i\in [1,n]}{\bm{\times}} \om^{\alpha'_{i,t(i)}}$. Then $f(s) = t \implies X_s \equiv X_t$. Therefore for every $t$ there exists $a_t\in\mathbb{N}$ such that
$$\underset{s \in Gr(X), f(s) = t}{\bigoplus} \w(X_s) = \w(X_t) \otimes a_t\;.$$

Note: this $\otimes$ can be replaced by $\cdot$ the usual product.

We want to compute $a_t = \bigl\{ s \in Gr(X)|f(s) = t\bigr\}$. First observe that $\bigl\{ s \in  Sl(X) |f(s) = t\bigr\} = \prod_{i\in [1,n]} a_{i,t(i)}$.
If $f(s) = t$ and $s \in Gr(X)$, then $t \in Gr'(X)$. This implies that $s$ is amongst the slices that are null in at least one of the $i$ where $t$ is null. Hence:
\begin{align*}
a_t &= \left(\prod_{t(i)=0} a_{i,0} - \prod_{t(i)=0} (a_{i,0} -1)\right) \cdot \prod_{t_i>0} a_{i,t(i)}\;,\\
\shortintertext{and}\quad \w(X) &= \underset{t \in Gr'(X)}{\bigoplus} \w(X_t)\otimes a_t\;.
\end{align*}

To go from developed normal form to short normal form, we can regroup together all slices $t$ in meta-slices $M_I\subseteq Sl'(X)$ indexed by $I \subseteq [1,n]$, such that $t \in M_I$ if for all $i\in [1,n]$,  $i \in I \Leftrightarrow t(i)=0$. Note that $M_{\emptyset}$ contains no grounded slices.

Let $X_I = \left(\underset{i\in I}{\bm{\times}}\om^{\alpha_i'}\right) \times \Bigl(\underset{i\not\in I}{\bm{\times} \sigma_i}\Bigr)$. Observe that, contrary to $X_s$ and $X_t$, $X_I$ is not a product of indecomposable ordinals. Our goal is to express $\w(X)$ as a natural sum of $\w(X_I)$:

\begin{align*}
\w(X)
&= \underset{I\neq\emptyset}{\bigoplus}\;\left(\underset{t\in M_I\cap Gr'(X)}{\bigoplus} w(X_t) \otimes a_t\right)\\
&= \underset{I\neq\emptyset}{\bigoplus}\;\left(\underset{t\in M_I\cap Gr'(X)}{\bigoplus} w(X_t) \otimes \prod_{t_i>0} a_{i,t(i)}\right)
\otimes \left(\prod_{i\in I} a_{i,0} - \prod_{i\in I} (a_{i,0} -1)\right)\\
&= \underset{I\neq\emptyset}{\bigoplus}\;\w(X_I)
\otimes L_I\;, \shortintertext{with} L_I &= \prod_{i\in I} a_{i} - \prod_{i\in I} (a_{i} -1)\;.
\end{align*}

Therefore we know three ways to express the width of a cartesian product of $n$ infinite ordinals, depending on the normal form in which they are written.

\begin{theorem}[Cartesian product of infinite ordinals]
\begin{align*}
\w(X) 
&= \bigoplus_{s\in Gr(X)} &&\w(X_s)\\
&= \bigoplus_{t \in Gr'(X)} &&\w(X_t) \otimes \left(\prod_{t(i)=0} a_{i,0} - \prod_{t(i)=0} (a_{i,0} -1)\right) \cdot \prod_{t_i>0} a_{i,t(i)} \\
&= \bigoplus_{I\subseteq[1,n],I\neq\emptyset}&&\w(X_I) \otimes \left(\prod_{i\in I} a_{i} - \prod_{i\in I} (a_{i} -1)\right)\;.
\end{align*}

\end{theorem}

\section{Combining finite and infinite ordinals}
\label{section-finite-infinite}

\begin{theorem}
\label{thm-product-finite-infinite}
For $X$ a cartesian product of infinite ordinals, and $k_1,\dots,k_m \in \mathbb{N}$, 
$$\w(X\times k_1 \times \dots \times k_m) = \w(X) \otimes k_1 \otimes \dots \otimes k_m$$
\end{theorem}

\begin{proof}
$X\times k_1 \times \dots \times k_m \geqaug X\times \Gamma_{k_1} \times \dots \times \Gamma_{k_m}$ so by Lemma \ref{lem-times-Gamma} we get the bound $\w(X\times k_1 \times \dots \times k_m) \leq \w(X) \otimes k_1 \otimes \dots \otimes k_m$.

To prove the other direction, first observe that $X\times k_1 \times \dots \times k_m \leqaug X\times (k_1\cdot \; \dots\; \cdot k_m)$. Therefore we only need to prove that $\w(X\times k)\geq \w(X)\otimes k$ for any $k \in \mathbb{N}$.

\medskip
We will adapt the proof of the lower bound of Theorem \ref{limit-formula}, keeping in mind the notations $\shift,\xi, Sl(X), Gr(X), \dots$

Let $Z \eqdef X \times k$. For all $s \in Sl(X)$, for all $r<k$ , we define the slice $Z_{s}^r \eqdef X_s \times \{r\}$.

Let $s_1,\dots,s_L$ be the grounded slices of $X$ ordered as in the proof of the lower bound of Theorem \ref{limit-formula}: as before $(X_{s_i})_{i\in[1,L]}$ is a quasi-incomparable family of subsets of $X$, and $\w(X_{s_L})+\dots + \w(X_{s_1}) = \bigoplus_{s \in Gr(X)} \w(X_s)$. 

For any $j=q\cdot k + r < k\times L$ with $q<L,r<k$, let $Z_j \eqdef Z_{s_{q + 1}}^{k-1-r}$: it boils down to ordering all the slices $Z_s^r$ for $s\in Gr(X)$ and $r<k$ such that $$Z_0,\dots,Z_{(k\times L) - 1} =Z_{s_1}^{k-1},Z_{s_1}^{k-2},\dots,Z_{s_1}^0,Z_{s_2}^{k-1},\dots,Z_{s_L}^0\;.$$

\medskip

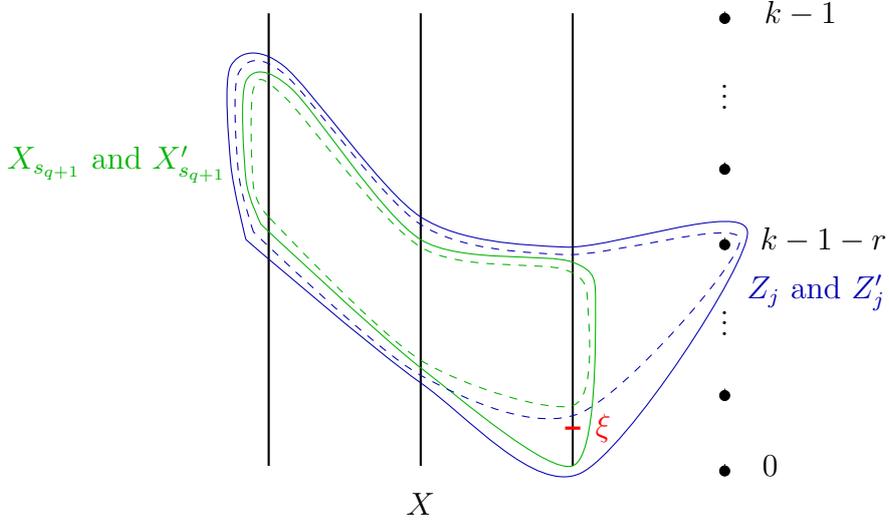
\begin{figure}[h]
\begin{centering}
                \begin{tikzpicture}
                	
                	\node (B) at (2, -.5){$X$};
                	\node [green!70!black](Xt) at (-2, 4) {$X_{s_{q+1}} \text{ and } X'_{s_{q+1}}$};
                	\node [blue!70!black](Xt) at (7.2, 2.3) {$Z_{j} \text{ and } Z'_{j}$};
                    \draw [thick](0,0) -- (0,6);
                    \draw [thick](2,0) -- (2,6);
                    \draw [thick](4,0) -- (4,6);
                    \draw [ultra thin, *-] (6,0) -- (6,0);
                    \draw [ultra thin, *-] (6,1) -- (6,1);
                    \draw [ultra thin, *-] (6,3) -- (6,3);
                    \draw [ultra thin, *-] (6,6) -- (6,6);
                    \draw [ultra thin, *-] (6,4) -- (6,4);
                    \node (z) at (6.6,0){$0$};
                    \node (k) at (7,6){$k-1$};
                    \node (k-) at (7.3,3){$k-1-r$};
                    \node (dots) at (6,2){$\vdots$};
                    \node (dots) at (6,5){$\vdots$};
                    \draw [green!70!black, dashed] plot[smooth]coordinates{(0, 3.3)(-.2,3.7)(-.2,5)(.2,4.9)(2,2.9)(3.9,2.6)(4.2,2.1)(4,0.8)(2,1.4)(0,3.3)};
                    \draw [green!70!black] plot[smooth]coordinates{(-.1, 3.2)(-.3,3.8)(-.3,5.1)(.3,5)(2,3)(4,2.7)(4.3,2.1)(4,0)(2,1.3)(-.1,3.2)};
                    \draw [blue!70!black, dashed] plot[smooth]coordinates{(-.2, 3.1)(-.4,3.9)(-.4,5.2)(.2,5.2)(2,3.2)(3.9,2.8)(6.2,3)(4.1,0.7)(2,1.2)(-.2, 3.1)};
                    \draw [blue!70!black] plot[smooth]coordinates{(-.3, 3)(-.5,4)(-.5,5.3)(.2,5.3)(2,3.3)(3.9,2.9)(6.3,3.1)(4.1,-.1)(2,1.1)(-.3, 3)};
                    \node [red](tresh) at (4.4,0.5){$\xi$};
                    \draw [red, very thick](3.9,0.5) -- (4.1,0.5);
               \end{tikzpicture}
               %remove the next empty line and centering doesnt work. latex sucks.

               \end{centering}
               \caption{$X_{s_{q+1}}$ in green, $X'_{s_{q+1}}$ in dashed green, $Z_j$ in blue, $Z'_j$ in dashed blue, for $j= k\cdot q +r$. }
        \end{figure}

Let us prove that $(Z_j)_{j<L\cdot k}$ is a quasi-incomparable family of subsets of $Z$.
Fix some  $j=k\cdot q + r$. Let $Y\subset Z_0 \cup \dots\cup Z_{j-1}$ a finite set. Let $Y_{|X}=\{x\in X|(x,r)\in Y\}$.
 We define $Z_j' \eqdef X'_{s_{q+1}}\times \{k-1-r\}$ where $X'_{s_{q+1}}=\shift(X_{s_{q+1}},Y_{|X})$ again. It is sufficient to prove that for all $h<j$, $Z'_j \perp (Y\cap Z_{h})$. Let $h = k \cdot q' + r'$ with $r'<k$. Either $q'<q$, or $q'=q$ and $r'<r$:
 \begin{itemize}
 \item If $q' < q$, follow the same reasoning as in the proof of the lower bound of Theorem \ref{limit-formula} to show that $X'_{s_{q+1}} \perp (Y_{|X}\cap X_{s_{q'+1}})$.
 \item If $q'=q$ and $r'<r$, then $k-1-r'>k-1-r$ so all elements of $Z_h$ are above all the elements of $Z_j$ in the last component. Since $s_{q+1}$ is grounded, there exists $i$ such that $s_{q+1}(i) = 0$ and the $i$th component of every elements of $Y\cap Z_h$ is below $\xi(Y,i)$, so below the $i$-th component of elements of $Z'_j$.
 \end{itemize}
   Thus $Z'_j \perp (Y\cap Z_h)$ for any $h<j$, so $Z'_j\perp Y$.

Therefore $(Z_j)_{j<L\cdot k}$ forms a quasi-incomparable family of subsets of $Z$, so according to Lemma \ref{lem-Antoine-strat}, $\w(Z) \geq \w(Z_{(L\cdot k) -1}) + \dots + \w(Z_0)$.

Observe that for any $j=k\cdot q +r$, $Z_j\equiv X_{q+1}$ and $\w(Z_j)=\w(X_{q+1})$ is indecomposable. So $\w(Z_{k\cdot q})+ \w(Z_{(k\cdot q) + 1})+\dots + \w(Z_{(k\cdot q) + k -1}) = \w(X_{q+1})\otimes k$. Therefore $\w(Z) \geq \w(X_{s_L})\otimes k \nfp \dots \nfp \w(X_{s_1})\otimes k = \w(X) \otimes k$.

\end{proof}

\section{Applications}
\label{section-w-equals-o}

\subsection{When width coincides with maximal order type}
In view of $\w((\omom)^{\times n}) = \o((\omom)^{\times n})$ (Proposion \ref{prop-omom-times-n}), one wonders if more generally $\w(X)$ catches up with $\o(X)$, for instance when the $\alpha_i$s are large enough? It turns out that we can exactly characterize the cartesian products $X$ such that $\w(X)$ and $\o(X)$ coincide:

\begin{theorem}
\label{thm-w-equals-o}
Let $Z=\alpha_1\times\dots\times \alpha_n \times k_1 \times \dots \times k_m$ with $n>0$, $\alpha_1,\dots,\alpha_n\geq\om$, and $0<k_1,\dots,k_m<\om$.
Now $\w(Z)=\o(Z)$ iff there exist: \begin{itemize}
\item $i\in [1,n]$ such that $\alpha_i$ is infinite indecomposable, and
\item $j_1 \neq j_2 \in [1,n]$ such that the Cantor normal forms of $\alpha_{j_1}$ and $\alpha_{j_2}$
only have infinite exponents (i.e., $\alpha_{j_1}$ and $\alpha_{j_2}$ are exactly divisible by $\omom$).
\end{itemize}
\end{theorem}

Note that $i$ can be equal to $j_1$ or $j_2$ (for instance in case $n=2$).

\begin{proof}

 According to Theorem \ref{thm-product-finite-infinite}, 
$\w(Z) = \w(X) \otimes k_1\otimes \dots \otimes k_{m}$
where $X=\alpha_1 \times \dots\times\alpha_n$ as before,
and $\o(Z) = \o(X) \otimes k_1\otimes \dots \otimes k_{m}$ according to Lemma \ref{lem-product-o}.
Therefore $\w(Z) = \o(Z)$ iff $\w(X) = \o(X)$.
\medskip 

$(\Rightarrow)$ Assume $\w(X) = \o(X)$.

We express $\o(X)$ in a form that allow us to compare it easily to $\w(X)$:
\begin{align*}
\o(X) &= \underset{i\in[1,n]}{\bigotimes} \alpha_i   \text{ according to Lemma \ref{lem-product-o}}\\
&= \underset{s\in Sl(X)}{\bigoplus}
\left(\underset{i\in[1,n]}{\bigotimes} \om^{\alpha_{i,s(i)}}\right) \text{ by distributivity}\\
&= \underset{s\in Sl(X)}{\bigoplus} \o(X_s )\;,\\
\shortintertext{and}\w(X) &= \underset{s\in Gr(X)}{\bigoplus} \w(X_s) \text{ according to Theorem \ref{limit-formula}.} 
\end{align*}

According to Lemma \ref{lem-w-leq-o}, for every slice $s\in Sl(X)$, $0<\w(X_s )\leq \o(X_s )$. Moreover $Gr(X)\subseteq Sl(X)$.
Therefore $\w(X) = \o(X)$ if and only if $Gr(X) =Sl(X)$ and $\w(X_s) = \o(X_s)$ for any $s\in Sl(X)$.
\begin{itemize}
\item $Gr(X) =Sl(X)$ implies that there are no ungrounded slices, i.e., there exists $i\in[1,n]$ such that $l_i = 1$. Thus there exists $i$ such that $\alpha_i$ is indecomposable. 
\item According to Theorem \ref{thm-eta}, $\w(X_s )= \o(X_s)$  is true iff there exist $j_1\neq j_2$ such that $\alpha_{j_1,s(j_1)}$ and $\alpha_{j_2,s(j_2)}$ are both infinite. In particular, for the top slice $s:j\mapsto l_j - 1$, there exist $j_1\neq j_2$ such that $\alpha_{j_1,l_{j_1}-1}$ and $\alpha_{j_2,l_{j_2}-1}$ are both infinite, and therefore all exponents of $\alpha_{j_1}$ and $\alpha_{j_2}$ are infinite.
\end{itemize}

$(\Leftarrow)$ The above proof makes it clear that the necessary conditions are sufficient.
\end{proof}

\subsection{Measuring elementary wqos}

Let the family of \emph{elementary} wqos be the smallest family of wqos that contains $\emptyset$  and is closed by disjoint sum, cartesian product and building finite sequences, as defined in \cite{SS-icalp11}. The maximal order type and height of any elementary wqo are already well-known (see Lemmas \ref{lem-product-o} and \ref{lem-product-h} for the cartesian product, \cite{dzamonja2020} for the disjoint sum, \cite{schmidt79} for finite sequences). Here we will show how to compute their width.
\begin{remark}
This family contains $\emptyset^*$ which is isomorphic to the singleton $\Gamma_1$, and $(\emptyset^*)^*$ which is isomorphic to $\om$. Since it is closed by disjoint sum, it contains also $\Gamma_k$ modulo isomorphism for all $k\in\mathbb{N}$.
\end{remark}

We can easily compute the width of a disjoint sum of wqos $A,B$: $\w(A\sqcup B)= \w(A)\oplus \w(B)$ according to Lemma \ref{lem-DSS-sqcup}. Moreover, observe that the cartesian product distributes over the disjoint sum: $A \times (B\sqcup C) = (A\times B) \sqcup (A\times C)$. Therefore we can restrict our study of the width to elementary wqos of the form $A_1^*\times \dots \times A_n^*$ with $A_1,\dots , A_n$ elementary wqos.

\medskip

Let $A$ be a wqo. Then the poset $A^*$ (also written $A^{<\om}$) of finite sequences on $A$ ordered by embedding is a wqo when $A$ is (\cite{higman52}).
Recall from \cite{schmidt79} and \cite{dzamonja2020}:
\begin{lemma}
\label{lem-words-o}

\begin{align*}
\o(A^*) =\w(A^*)=\begin{cases}
\om^{\om^{\o(A) -1}} &\text{ if $\o(A)$ is finite,}\\
\om^{\om^{\o(A) +1}} &\text{ if $\o(A) = \delta + n$ with $\om^{\delta}=\delta$ and $n$ finite,}\\
\om^{\om^{\o(A)}} &\text{ otherwise.}
\end{cases}
\end{align*}
We write that in the simpler form $\o(A^*)=\om^{\om^{\o(A)'}}$.
\end{lemma}

\begin{remark}
If $A\neq\emptyset$ then $A^*$ is transferable.
\end{remark}
\begin{proof}
For any $u\in A^*$, for any $a\in A$, $A^*_{\not\leq u}$ contains $\{uav|v\in A^*\}$, which is isomorphic to $A^*$. Therefore $A^*$ is self-residual, hence transferable.
\end{proof}

\begin{remark}
By Lemma \ref{lem-words-o}, if $\o(A)>1$, then $\o(A^*)$ verifies the two conditions described in Theorem \ref{thm-w-equals-o}: $\o(A^*)$ is infinite indecomposable, and its normal form only have infinite exponents. 
\end{remark}

This property of $A^*$ will prove useful thanks to the following theorem, which generalises Theorem \ref{thm-w-equals-o} to the cartesian products of $n$ wqos:

\begin{theorem}
\label{thm-app}
Let $A_1,\dots,A_n$ be a family of wqos.
If there exist $i,j_1\neq j_2 \in [1,n]$ such that $\o(A_i)$ is infinite indecomposable, and $\o(A_{j_1})$ and $\o(A_{j_2})$ only have infinite exponents, then $\w(A_1\times\dots\times A_n) = \o(A_1) \otimes \dots \otimes \o(A_n)$.
\end{theorem}

\begin{proof}
According to Lemma \ref{lem-w-leq-o},
$$ \w(A_1\times\dots\times A_n)\leq \o(A_1) \otimes \dots \otimes \o(A_n)\;.$$

Since $A_i \leqaug \o(A_i)$ for any $i$, $A_1\times\dots\times A_{n}\leqaug \o(A_1) \times \dots \times \o(A_n)$, thus:
\begin{align*}
\w(A_1\times\dots\times A_n) &\geq \w(\o(A_1) \times \dots \times \o(A_n))\\
&= \o(\o(A_1) \times \dots \times \o(A_n)) \text{ according to Theorem \ref{thm-w-equals-o},}\\
&= \o(A_1) \otimes \dots \otimes \o(A_n)\;.\qedhere
\end{align*}

\end{proof}

Let $X=A_1^*\times\dots\times A_n^*$ with $n\geq 2$ and $A_i\neq\emptyset$ elementary for all $i\leq n$. Let's compute $\w(X)$:\begin{itemize}
\item If there exist $i\neq j\in[1,n]$ such that $\o(A_i)>1$ and $\o(A_j)>1$, then the conditions of Theorem \ref{thm-app} are fulfilled, and $\w(X)=\o(X)$.
\item Otherwise if there exists $i\in [1,n]$ such that $\o(A_i)>1$ and for all $j\neq i$, $o(A_j)=1$, i.e., $A_j \equiv\Gamma_1$ and $A_j^* \equiv \om$, then $X\equiv A_i^*\times \om^{\times (n-1)}$, and $\w(X)=\o(X)$. Lemma \ref{lem-w-leq-o} gives us the upper bound, and since $A_i^*$ is transferable:

\begin{align*}
\w( X^*\times\om^{\times n})&\geq \w(X^*)\cdot \o(\om^{\times n}) &\text{ according to Lemma \ref{lbiftransferable},}\\
&= \om^{\om^{\o'(X)} + n} &\text{ according to Lemma \ref{lem-words-o},}\\
&=\o(X^*)\otimes \o(\om^{\times n})\\
&= \o(X^*\times\om^{\times n})\;,
\end{align*}
which we know how to compute with Lemma \ref{lem-product-o}.
\item Otherwise, $A_i\equiv\Gamma_1$ for all $i\in [1,n]$, hence $X\equiv \om^{\times n}$. According to Proposition \ref{prop-om-times-n}, $\w(\om^{\times n})=\om^{n-1}$ for all $n\geq 1$.
\end{itemize}

Therefore we can measure any elementary wqos for all ordinal invariants.

\section{Product of finite ordinals}
\label{section-finite}

In the case of the cartesian product of finite ordinals, we have a finite poset, thus its width coincide with the length of its largest antichain. For the sake of completeness, we recall a classic result that characterize its width.

Let $k_1,\dots, k_n>0$ be $n$ finite ordinals, and $p_1,\dots,p_n$ distinct prime numbers. Observe that $X\eqdef k_1\times\dots\times k_n$ is isomorphic to the poset of the divisors of $p_1^{k_1-1}\cdot\;\dots\;\cdot p_n^{k_n-1}$ ordered by divisibility.
Therefore, according to Theorem 1 of \cite{debruijn51b}:

\begin{theorem}
\label{thm-finite}
One maximal antichain of $X$ is:
$$A = \left\{(m_1,\dots,m_n)\in X\;\middle|\; \sum m_i = \left\lfloor\frac{1}{2}\sum (k_i-1)\right\rfloor \right\}\;,$$
therefore $\w(X)=|A|$.
\end{theorem}

For instance, 
\begin{align*}
w(2^{\times n})&=\left|\left\{(m_1,\dots,m_n)\in \{0,1\}^{\times n}\;\middle|\; \sum m_i = \left\lfloor\frac{n}{2}\right\rfloor \right\}\right|\\
&= \binom{n}{\big\lfloor\frac{n}{2}\big\rfloor} \text{ the $n$th central binomial coefficient.}
\end{align*}

Similarly, $w(3^{\times n})$ is equal to the central trinomial coefficient, defined as the largest coefficient of the polynomial $(1+x+x^2)^n$. We can compute it efficiently:
$$\w(3^{\times n})=\sum_{0\leq i \leq \left\lfloor {n}/{2}\right\rfloor}\binom{n}{i}\binom{n-i}{i}\;.$$
This leads to a slightly different characterization of $\w(X)$ which can be deduced from Theorem \ref{thm-finite}:

\begin{corollary}
For $X=k_1\times\dots\times k_n$ a cartesian product of finite ordinals, $\w(X)$ is equal to the central coefficient of the polynomial $P_{k_1}\times\dots\times P_{k_n}$, where $P_{k_i}(x)=1 + x + \dots + x^{k_i -1}$.
\end{corollary}

\begin{proof}
The central coefficient of $P_{k_1}\times\dots\times P_{k_n}$ is the coefficient of $x^{\left\lfloor\left(\sum k_i-1\right)/2\right\rfloor}$.
\end{proof}

\section{Conclusion}

Following \cite{kriz90b} and \cite{dzamonja2020}, we consider the width of wqos. Our work addresses the issue of computing the width of the cartesian product of known wqos, more specifically of $n$ ordinals, extending \cite{abraham87} which solved the case $n=2$. 

 Together, Theorems \ref{thm-eta-Z}, \ref{limit-formula}, \ref{thm-product-finite-infinite}, and \ref{thm-finite} cover all the cases needed to compute the width of the cartesian product of finitely many ordinals.
 
 These theorems rely on the well-known method of residuals to prove upper bounds, a game-theoretical approach to prove lower bounds, and new techniques to transfer strategies from simple wqos to more complex ones.
 
 Beyond the cartesian product of linear orders, our result can be used to compute the width of a generic family of elementary wqos.

The techniques developed here can help target other open questions on wqo width, for example how to compute the width of the powerset or the set of multisets over known wqos.

\section*{Aknowledgements}

The research described in this article started with my master internship supervized by Ph.~Schnoebelen. It has also benefited from numerous discussions and suggestions from M.~Dzamonja, S.~Halfon,  and S.~Schmitz.

% References
%======================
\bibliographystyle{alpha}
\bibliography{biblio}

\end{document}